\documentclass{article}

\usepackage[utf8]{inputenc}
\usepackage[T1]{fontenc}
\usepackage[sort]{natbib}

\usepackage{amsmath}
\usepackage{amssymb}
\usepackage{amsthm}
\usepackage{amsfonts}
\usepackage{fullpage}
\usepackage{thm-restate}
\usepackage{enumitem}
\usepackage{subcaption}

\usepackage[linesnumbered, boxed, vlined]{algorithm2e}
\DontPrintSemicolon
\SetKwInOut{Input}{Input}
\SetKwInOut{Output}{Output}
\SetKwProg{function}{Function}{:}{}

\usepackage{algorithmic}

\usepackage{graphicx}
\usepackage{enumitem}
\usepackage{textgreek}
\usepackage{tabularx}
\usepackage{xspace}
\usepackage[linktoc=section]{hyperref}
\usepackage[nameinlink,capitalize,nosort]{cleveref}
\usepackage[dvipsnames,svgnames]{xcolor}
\usepackage{tcolorbox}

\definecolor{ForestGreen}{rgb}{0.0333,0.4451,0.0333}
\definecolor{DarkRed}{rgb}{0.65,0,0}
\definecolor{Red}{rgb}{1,0,0}
\hypersetup{
    unicode=false,          
    colorlinks=true,        
    linkcolor=BrickRed,          
    citecolor=OliveGreen,        
    filecolor=magenta,      
    urlcolor=cyan           
}

\newtheorem{theorem}{Theorem}[section]
\newtheorem{corollary}[theorem]{Corollary}
\newtheorem{lemma}[theorem]{Lemma}

\newtheorem{definition}[theorem]{Definition}
\newtheorem{observation}{Observation}[section]

\makeatletter
\newcommand{\multiline}[1]{%
  \begin{tabularx}{\dimexpr\linewidth-\ALG@thistlm}[t]{@{}X@{}}
    #1
  \end{tabularx}
}
\makeatother

\newcommand{\PTime}{\textsc{P}\xspace}

\newcommand{\polylog}{{\ensuremath\mathsf{polylog}}}
\newcommand{\NC}{\ensuremath{\mathsf{NC}}\xspace}
\newcommand{\CC}{\ensuremath{\mathsf{CC}}\xspace}
\newcommand{\RNC}{\ensuremath{\mathsf{RNC}}\xspace}
\newcommand{\ef}{\ensuremath{\mathsf{EF}}\xspace}
\newcommand{\efone}{\ensuremath{\mathsf{EF1}}\xspace}
\newcommand{\efk}[1]{\ensuremath{\mathsf{EF}}$#1$}
\newcommand{\efx}{\ensuremath{\mathsf{EFX}}\xspace}
\newcommand{\po}{\ensuremath{\mathsf{PO}}\xspace}
\newcommand{\efonepo}{\ensuremath{\mathsf{EF1\cap PO}}\xspace}
\newcommand{\APX}{\ensuremath{\mathsf{APX}}\xspace}
\newcommand{\NP}{\ensuremath{\mathsf{NP}}\xspace}
\newcommand{\PPAD}{\ensuremath{\mathsf{PPAD}}\xspace}

\newcommand{\defn}[1]{\emph{\textbf{#1}}}

\renewcommand{\emptyset}{\varnothing}

\def \eps {\varepsilon\xspace}

\def \cM {\mathcal{M}\xspace}
\def \cG {\mathcal{G}\xspace}
\def \cA {\mathcal{A}\xspace}
\def \cP {\mathcal{P}\xspace}

\def \cI {\mathcal{I}\xspace}
\def \gap {\cG\xspace}

\title{Improved Parallel Algorithms for EF1 Allocations}

\author{Kishen N Gowda\thanks{Supported in part by NSF grant CNS-2317194} \\ University of Maryland\and D Ellis Hershkowitz\thanks{Supported in part by NSF grant CCF-2403236 and NSF CAREER Award CCF-2540331.}\\ Brown University \and Richard Z Huang\footnotemark[2]\\Brown University \and Gregory Kehne\\Washington University in St.\ Louis}
\date{}

\begin{document}

\maketitle

\begin{abstract}
Allocating $m$ indivisible goods among $n$ agents is a fundamental task in fair division. Recent work of Garg and Psomas [AAMAS 2025] initiated the study of parallel algorithms for envy-free up to one good (\efone) allocations, giving \NC algorithms for $2$ and $3$ agents. They also showed \CC-hardness results for simulating the classic Round Robin algorithm for \efone allocations, even when each agent values at most $3$ goods and each good is valued by at most $3$ agents.

We strengthen these results. For the case of $2$ agents, we quadratically improve the depth from $O(\log ^ 2 m) $ to $O(\log m)$ and the work from $O(m \log m)$ to $O(m)$. Furthermore, we significantly generalize beyond $3$ agents by giving \NC algorithms for any constant number of agents.  We also give randomized algorithms with depth $\tilde{O}(m/n)$ and polynomial work. As corollaries of these results, we obtain \NC algorithms whenever each agent values at most $\polylog(m)$ goods and each good is valued by at most $O(1)$ agents, and \RNC algorithms when each agent values at most $\polylog(m)$ goods. As such, our algorithms bypass the \CC-hardness of Garg and Psomas by not simulating Round Robin. We also complement the aforementioned \CC-hardness by showing the \CC-completeness of simulating Round Robin. Lastly, beyond \efone allocations, we show that computing envy-free up to $k$ goods allocations is possible for $k \approx \sqrt{m}$ in \RNC, or $k = m^{\eps}$ in sublinear depth for any constant $\eps > 0$. 
\end{abstract}
\pagenumbering{gobble}

\newpage
\pagenumbering{arabic}
\setcounter{page}{1}

\section{Introduction}


The fair allocation of indivisible goods is a central problem in computational social choice \cite{brandt16handbook}. Here, we are given $m$ goods and $n$ agents and must decide how to assign goods to agents in such a way as to satisfy some notion of fairness. One natural such notion is \defn{envy-freeness} (\ef): no agent $a$ should prefer another agent $b$'s bundle of goods over their own \cite{foley66resource}. Unfortunately envy-freeness is unobtainable, as illustrated by the case of two agents and one good. A classic relaxation of envy-freeness which circumvents this issue is envy-freeness-up-to-one-good (\efone): no agent $a$ should prefer another agent $b$'s goods over their own as long as they are allowed to ignore their favorite of $b$'s goods \cite{budish11combinatorial,lipton04approximately}.

A simple way of seeing that \efone allocations always exist is a proof by algorithm: in particular, the classic Round Robin protocol. Here, we fix an ordering on the agents and then, in this order, each agent repeatedly chooses their favorite among the remaining goods. Assuming agents' values for goods are additive---that is, agent $a$ gives good $g$ value $v_a(g)$ and gives goods $S$ value $v_a(S) := \sum_{g \in S} v_a(g)$---then it is a simple exercise to see that the resulting allocation is \efone. Motivated by the extensive work on algorithms for \efone allocations and the scale of many \efone instances of practical interest, \cite{garg2023fairly} initiated the study of parallel algorithms for \efone allocations.  

Efficient parallel algorithms for \efone allocations based on Round Robin seem quite hopeless. In particular, an efficient parallel algorithm is generally understood to be an \NC algorithm---namely, one that has polylogarithmic depth and polynomial work.\footnote{We make use of standard work-depth models of parallel computation; see Section~\ref{sec:prelims} for details.}  Round Robin, however, induces chains of linear-length sequential dependencies since each good chosen by each agent delicately depends upon all prior choices of all other agents. Formalizing these issues, \cite{garg2023fairly} showed that simulating Fixed-Order Round Robin---where the ordering of the agents is part of the input---is Comparator-Circuit-hard (\CC-hard). \CC-hardness is a standard notion of parallel hardness widely thought to rule out \NC algorithms \cite{mayr92complexity,cook14complexity}. Even worse, this \CC-hardness holds in the case of additive valuations, even when each agent has non-zero value for at most $3$ goods and each good has at most $3$ agents that give it non-zero value.

While several other efficient algorithms for \efone allocations are known, including envy-cycle elimination \cite{lipton04approximately} and recursive probabilistic serial \cite{aziz24best}, like Round Robin these algorithms allocate goods one-at-a-time and so appear fundamentally sequential. 
Still other approaches appear computationally intractable, optimizing \NP-hard \cite{plaut20almost} and even \APX-hard objectives \cite{caragiannis19unreasonable,lee2017apxnsw} or finding equilibria in games for which the task is \PPAD-complete \cite{budish11combinatorial,othman16complexity}.

That successful approaches for \efone allocations are sequential may not be an accident; \efone is a global property of instances in that changing the allocation of a single good can dramatically affect envy relationships between the donor, the recipient and other agents. Likewise, it is not clear how to efficiently combine two disjoint \efone allocations of goods with a common set of agents into a single allocation---a major barrier to divide-and-conquer-type parallel algorithms.

The apparent difficulty of computing \efone allocations in parallel for general instances motivates our central question:
\begin{quote}\centering
    \textit{Does there exist a reasonable class of \efone allocation instances\\ which admit efficient parallel algorithms?}
\end{quote}
\noindent Instances with a small number of agents is among the most promising such classes. In particular, there is extensive work on algorithms for fair division with boundedly-many agents, including on the query complexity of \efone allocations for $2$ or $3$ agents \cite{oh21fairly}, the comparison-based query complexity for \efone allocations \cite{bu24logarithmic} for $2$ or $3$ agents, \efone allocations that are Pareto optimal for constantly-many agents \cite{mahara2025polynomial} and envy-free up to any good (\efx) allocations for $2$ \cite{plaut20almost} or $3$ agents \cite{chaudhury24efx}; we give more details on these works in \Cref{sec:relwork}. In the spirit of these works, \cite{garg2023fairly} recently gave \NC algorithms with $O(\log ^2m)$ depth and $O(m \log m)$ work for $2$ or $3$ agents with additive valuations. However, no \NC algorithms are known for \efone allocations for even only $4$ agents.

\subsection{Our Contributions}
In this work, we significantly improve on the existing parallel algorithms for \efone allocations.

First, for the case of two agents, we quadratically improve the $O(\log ^ 2 m)$ depth of \cite{garg2023fairly} to $O(\log m)$ while also improving the work from $O(m \log m)$ to $O(m)$.



\begin{restatable}{theorem}{twoagents}\label{thm:2-agents}
    There is an $O(\log m)$ depth and $O(m)$ work algorithm that outputs an \efone allocation for two agents with additive valuations.
\end{restatable}
\noindent

\noindent Our algorithm for the two-agent case can be understood as implementing a certain recursive game. An instance of the game comes with an ordering over the agents and, in turn, decomposes into two subgames. In each game, the agent that goes first chooses one subgame in which to go first and one subgame in which to go second. The base case consists of the agent that goes first and chooses their favorite of two goods, leaving the other for the agent that goes second. The key to making this work is carefully choosing and reasoning about the rule by which agents choose the subgame in which they go first---importantly, agents do not just choose the subgame that allows them to maximize the utility they receive.

As a corollary of \Cref{thm:2-agents}, we obtain an equally efficient algorithm for previously studied  ``graph'' instances \cite{christodoulou23fair,afshinmehr25efx,sgouritsa25existence,bhaskar25extending,Misra2025EFGraph}, where each good is positively valued by at most two agents.

\noindent 

Next, we significantly expand the set of instances for which \NC algorithms are known by giving  an \NC  algorithm for \efone allocations for any constant number of agents (as opposed to the algorithm of \cite{garg2023fairly} for $n\leq 3$ agents). More specifically, we show that Fixed-Order Round Robin can be implemented in \NC for any constant number of agents.
\begin{restatable}{theorem}{reachabilityreduction}\label{thm:reachabilityreduction}
    There is an \NC algorithm that, given $O(1)$ agents with additive valuations and any ordering on these agents, outputs the \efone allocation of Round Robin with this ordering.
\end{restatable}
\noindent Our approach for the $O(1)$ agents case is roughly as follows. We view Round Robin as a deterministic process where each round corresponds to a particular \emph{configuration} of available goods. Then, we show that, given any configuration, the chosen good of the corresponding round and configuration of the next round of Round Robin can be computed efficiently in parallel. This allows us to construct a directed graph whose nodes are configurations and whose edges encode valid round-to-round transitions. Since Round Robin is deterministic, the configurations reachable from the initial state form a unique directed path. Recovering this path, together with the allocations stored along it, yields the final allocation.

The challenge in instantiating this approach is that, naively, the number of possible configurations can be exponentially large since the number of subsets of goods is exponential. This would yield a graph that is too large for an \NC algorithm. The key trick to overcome this is to observe that it is not possible for any subset of goods to be available in a given round of Round Robin; rather, the goods available in a given round of Round Robin can be uniquely recovered from each agent's favorite remaining good. Thus, we only need to construct a configuration for each possible arrangement of remaining favorite goods of agents and, as long as we have $O(1)$ agents, the number of such configurations is polynomially bounded.


Similar to our application of our two agent algorithm to graph instances, we show that our algorithm for $O(1)$ agents can be used as a subroutine to obtain \NC algorithms for certain sparse ``hypergraph instances'' whose edges correspond to goods. The basic idea of our hypergraph algorithm is to color the goods/edges of the hypergraph and then use \Cref{thm:reachabilityreduction} to allocate the goods of each color sequentially. As a special case of the hypergraph instances we solve, we obtain the following result.
\begin{restatable}{corollary}{bypassHard}\label{thm:bypassHard}
   There is an \NC algorithm that outputs an \efone allocation when valuations are additive, each agent positively values at most $\polylog(m,n)$ goods and each good is positively valued by at most $O(1)$ agents.
\end{restatable}
\noindent  Notably, this bypasses the aforementioned \CC-hardness of \cite{garg2023fairly} for simulating Round Robin (which holds when each agent values at most $3$ goods and each good is valued by at most $3$ agents). Thus, crucially, while \Cref{thm:reachabilityreduction} simulates Round Robin, when we use \Cref{thm:reachabilityreduction} as a subroutine for our sparse hypergraph instances we end up with a new algorithm which does not simulate Round Robin.

The above result shows that efficient parallel algorithms are possible whenever the number of agents is small. Our next result shows that efficient parallel algorithms are possible if the number of agents is large. In particular, we give depth $\tilde{O}(m/n)$ algorithms with polynomial work as described below.
\begin{restatable}{theorem}{mnsmall}\label{thm:m/n-small}
    There is a randomized algorithm that runs in depth $O\left(\frac{m\log^3 m}{n}\right)$ and work $O(m^{5.5}\log^3 m)$ and computes an \efone allocation with high probability. 
\end{restatable} 
\noindent The above result is based on repeatedly solving certain instances of maximum matching which, in turn, is known to admit \RNC algorithms \cite{Karp1986Constructing}. Similar to our application of \Cref{thm:reachabilityreduction} to obtain \Cref{thm:bypassHard},  we can apply \Cref{thm:m/n-small} to obtain the below corollary.
\begin{restatable}{corollary}{bypassHardAgain}\label{thm:bypassHardAgain}
   There is an \RNC algorithm that outputs an \efone allocation when valuations are additive, each agent positively values at most $\polylog(m,n)$ goods.
\end{restatable}
\noindent Notably this also bypasses the \CC-hardness of \cite{garg2023fairly} while only assuming a bound on the number of goods each agent positively values; however, it comes at the cost of using randomization.

We also complement the \CC-hardness of \cite{garg2023fairly} by showing that simulating Round Robin---in particular, Fixed-Order Round Robin where the ordering on agents is given as input---is not just hard for the class \CC but also complete for it.

\begin{restatable}{theorem}{cccomplete}\label{thm:forr-cc}
    Fixed-Order Round Robin is \CC-Complete.
\end{restatable}
\noindent The above result is based on a log-space reduction to the stable matching problem.


We next study what non-trivial parallel algorithms are possible if we relax beyond \efone. In particular, we study envy-free up to $k$ (\efk{\left(k\right)}) allocations for $k > 1$. First, we show that randomized \NC (\RNC) algorithms are possible if we allow for $k \approx \sqrt{m}$.
\begin{restatable}{theorem}{efrootm}\label{thm:efrootm} There is an \RNC algorithm with depth $O(\log n)$ and work $O(m)$ for additive valuations that outputs an \efk{\left(\tilde{O}\left(\sqrt{m/n}\right)\right)} allocation with high probability.\footnote{With high probability means at least $1- \frac{1}{n^c}$ for any constant $c$. Since \emph{verifying} an \efone allocation is in \NC \cite[Theorem 2]{garg2023fairly}, using standard approaches these \RNC algorithms can be adapted to Las Vegas algorithms that have only expected depth and work bounds, but always output \efone allocations.}
\end{restatable}
\noindent Our algorithm is based on exploiting concentration bounds for random walks. In particular, we randomly allocate our goods. In order to argue agent $a$ envies $b$ up to at most $\approx \sqrt{m}$ goods, we fix all randomness except for that corresponding to $a$ and $b$. We then use the remaining randomness to construct an appropriate random walk in which, as long as the walk does not drift more than $\approx \sqrt{m}$ from the origin, we know $a$ envies $b$ by at most $\approx \sqrt{m}$ goods.

Strengthening our fairness guarantee while weakening our runtime guarantee, we show that better-than-sequential---namely, sublinear depth---algorithms are possible for $k = m^\eps$ for any constant $\eps > 0$.
\begin{restatable}{theorem}{efm}\label{thm:efm}
    For any $\eps>0$, there is an $\Tilde{O}\left(m^{1-\eps}\right)$ depth and $O\left((nm)^{5.5}\log^3(nm)\right)$ work  randomized algorithm for additive valuations that outputs an \efk{\left(m^{\eps}/n\right)} allocation with high probability.
\end{restatable}
\noindent \Cref{thm:efm} follows by partitioning our goods into $\approx m^{\eps}$ piles and running the algorithm of \Cref{thm:m/n-small} on each pile: each agent envies each other agent by at most one good per pile, giving an \efk{\left(k\right)} allocation for $k \approx m^{\eps}$.



\subsection{Related Work}
\label{sec:relwork}
As described above, the most closely related work is that of \cite{garg2023fairly}, who initiated the study of fair division in parallel models of computation. 
In addition to the $n=2$ and $n=3$ agent results and \CC-hardness of Fixed-Order Round Robin described above, they also give \NC algorithms for \efone when agents' additive valuations collectively take on $O(1)$ distinct values, and for recognizing \efone and \efx allocations.
Finally, they consider the problem of computing an \ef allocation of $m$ goods to $n$ additive agents with \emph{minimum subsidy}, meaning a combined allocation of goods and payments that is \ef and uses the minimum total payment possible.

Beyond \cite{garg2023fairly}, the most closely related lines of work are in parallel algorithms and algorithmic fair division and allocation. We briefly survey both lines.

\paragraph*{Parallel Algorithms.}
As mentioned above, the complexity classes \NC and \RNC are problems in \PTime that efficiently parallelize; see \Cref{defn:nc} for the formal definition. We introduce several problems in \NC that will be of use to us in \Cref{sec:prelims}. By contrast, a problem is \PTime-complete if it is in \PTime and any problem in \PTime is reducible to it via a many-one logspace reduction; i.e. any problem in \PTime can be transformed into it in a highly parallel way. \PTime-completeness is widely thought to rule out \NC and \RNC algorithms \cite{greenlaw1995limits}. \CC is a sort of intermediate complexity class: $\CC \subseteq \PTime$, but it is open whether $\CC \subseteq \NC$ or $\NC \subseteq \CC$.
With that said, $\CC$-hardness is generally thought to rule out \NC algorithms \cite{subramanian1989new}; see \Cref{defn:cc} for the formal definition.

Beyond \efone allocations, other foundational problems and algorithms in social choice have been studied in the context of parallel algorithms.
Indeed, attempts to characterize the stable matching problem motivated the definition of \CC, for which it is complete \cite{subramanian1989new,mayr92complexity}.
More recently, properties of matchings under one-sided preferences were studied in the parallel setting by \cite{zheng19parallel,hu20nc}.  Likewise, parallel algorithms have been studied for several voting rules: early works showed that deciding winners of several prominent voting rules \cite{csar17winner,csar18computing} and membership in choice sets \cite{brandt09computational} are in \NC, but that others, including the single transferrable vote (STV), are \PTime-complete.
Multiwinner voting was also recently studied by \cite{fitzsimmons2025parallelizability}, who establish that prominent approval-based committee rules are \PTime-hard, but the problem is more tractable on restricted domains.



\paragraph*{Fair Allocation.}
As mentioned above, both Round Robin for additive valuations and Envy-Cycle Elimination for more general valuations establish that \efone allocations can be found in polynomial time.
Fine-grained questions about the difficulty of computing \efone allocations have received more recent attention. 
The \defn{query complexity} of computing \efone allocations, as measured by calls to agents' value functions, was shown to be $\Theta(\log m)$ for $n=2$ and $n=3$ agents with additive utilities \cite{oh21fairly}, and \cite{garg2023fairly} leverage this to show \efone is in \NC for 2 and 3 agents.
Also building on \cite{oh21fairly}, \cite{bu24logarithmic} show the \emph{comparison-based} query complexity of \efone allocations is $O(\log m)$ for $2$ and $3$ agents, but are able to extend this bound to any constant number of agents only for identical values.
The \defn{communication complexity} of finding \efone allocations was also recently shown to be $O(m \log m)$ by \cite{feige25low}.

There is also an active line of work on questions of whether \efone allocations can be provided in conjunction with other guarantees; most prominently the efficiency notion of Pareto Optimality (\po), which stipulates that no other allocation is preferred by all agents.
The existence of \efonepo allocations for additive valuations was established by \cite{caragiannis19unreasonable} via the allocation which maximizes Nash welfare, but finding this allocation is \APX-hard \cite{lee2017apxnsw}. 
More recently, constructing such \efonepo allocations was shown to be in pseudo-polynomial time \cite{barman18finding}, and in polynomial time for a constant number of agents \cite{mahara2025polynomial}. 

A distinct line of work focuses on fair allocations in structurally restricted settings, particularly in pursuit of the stronger fair allocation standard of \defn{envy-free up to any good} (\efx) \cite{caragiannis19unreasonable}. 
Such allocations have been shown to exist for all $2$- \cite{plaut20almost} and $3$-agent instances \cite{chaudhury24efx}.
Additionally, when valuations are \defn{graph-structured}, meaning that each good (edge) is valued by at most $2$ agents (vertices), \cite{christodoulou23fair} show that \efx allocations exist and can be found in polynomial time.
We present positive results for graph-structured instances in \cref{subsec:graph}, and for their constant-rank hypergraph generalizations in \Cref{subsec:3hypergraph}.
See \cite{amanatidis23fair} for a survey of recent work on related fairness notions in indivisible allocation.



\subsection{Preliminaries}
\label{sec:prelims}
This section introduces the necessary background, notation, and standard results used throughout the paper.

\paragraph*{Parallel Computation Model.}
We analyze the theoretical efficiency of our parallel algorithms in the binary fork-join model~\cite{blelloch2020optimal}, a concrete \defn{work-depth model} used to analyze many recent modern parallel algorithms.
The model is defined in terms of the two complexity measures \defn{work} and \defn{depth}~\cite{blelloch2020optimal, jaja1997introduction,CLRS}.
The \defn{work} is the total number of operations executed by the algorithm, and the \defn{depth} is the length of the longest chain of sequential dependencies.  Computations with work $W$ and depth $D$ can be executed using a randomized work-stealing scheduler in practice in $W/P+O(D)$ time with high probability on $P$ processors~\cite{BL98,ABP01}.
We also note that computations in this model can be cross-simulated in standard variants of the PRAM model with the same work (asymptotically), and losing at most a single logarithmic factor in the depth~\cite{blelloch2020optimal}.
%
\begin{definition}[\NC, \RNC]\label{defn:nc}
A problem is in \NC if it admits a deterministic parallel algorithm with polylogarithmic depth and polynomial work. A problem is in \RNC if it admits a randomized parallel algorithm with polylogarithmic depth and polynomial work that is correct with high probability.
\end{definition}
\begin{definition}[\CC]\label{defn:cc}
A problem is in \CC if it can be computed by a logspace-uniform family of polynomial-size comparator circuits.
Furthermore, a problem is said to be \CC-hard if any problem in \CC can be reduced to it via logspace reductions.
\end{definition}
\noindent
We next state several standard parallel primitives and algorithms that we will use.
\begin{lemma}[Parallel Filter;~\cite{jaja1997introduction}]\label{lemma:parallel_filter}
Given a sequence of $n$ elements and a predicate on these elements, there is a stable parallel filter that returns the subsequence of all elements satisfying the predicate, while preserving their relative order, in $O(\log n)$ depth and $O(n)$ work.
\end{lemma}
\noindent
We use \defn{parallel tree contraction}~\cite{miller1985parallel}, a generic framework that repeatedly contracts degree-$1$ and degree-$2$ vertices to reduce a rooted tree to a single node, while maintaining suitable summary information throughout the contractions. It can be used to support standard subtree and node-to-root path queries, typically recovered through a subsequent uncontraction phase. In particular, we use the following consequence in our constant-agent algorithm.
\begin{lemma}[Parallel Tree Contraction;~\cite{gazit1988optimal}]\label{lemma:treecontraction}
    Given a rooted tree on $n$ vertices, any node-to-root path query with an associative aggregate can be supported in $O(\log^2 n)$ depth and $O(n)$ work using parallel tree contraction.
\end{lemma}
\noindent
We will also use the following parallel maximal independent set and vertex coloring results, which is a key ingredient in extending the constant-agent algorithm to a more general class of instances.
\begin{lemma}[Parallel Maximal Independent Set (MIS); \cite{goldberg1989ParMIS}]\label{lem:det_par_mis}
    Given a graph $G$ with $n$ vertices and $m$ edges, a maximal independent set of $G$ can be computed in $O(\log^4n)$ depth and $O((m+n)\log^2n)$ work.
\end{lemma}
\begin{lemma}[Parallel Vertex Coloring]\label{lemma:coloring}
    Given a graph with $n$ vertices, $m$ edges, and maximum degree $\Delta$, a proper $(\Delta+1)$-vertex-coloring can be computed in $O(\Delta\log^4 n)$ depth and $O(\Delta(m+n)\log^2n)$ work.
\end{lemma}
\begin{proof}
    Repeatedly compute an MIS in the graph induced by the currently uncolored vertices, assign these vertices a new color, and remove them. Since each color class is an independent set, the resulting coloring is proper.
    Moreover, in any round, each remaining vertex is either selected in the MIS or has at least one neighbor that is selected. Thus, in every round that a vertex survives, its degree decreases by at least one, resulting in at most $\Delta+1$ rounds in total. The work and depth bounds follow by \Cref{lem:det_par_mis}.
\end{proof}
\noindent
We note that better work and depth bounds are known for several of the above primitives if one allows randomization. Here we use deterministic variants in order to establish our NC results. Next, we state the randomized primitives used in the paper. Some of our algorithms will require sampling uniformly random permutations and computing maximum matchings. We use the following results.
\begin{restatable}[Parallel Permutation Sampling; Section 6, \cite{blelloch2020optimal}]{lemma}{permutation}\label{lemma:permutation}
    There is an algorithm that runs in $O(\log n)$ depth and $O(n)$ expected work and samples a uniformly random permutation of $[n]$ with high probability.
\end{restatable}
\begin{theorem}[Parallel Maximum Matching;  \cite{Mulmuley1987Matching}]\label{thm:matching}
    Given a graph $G$ with $n$ vertices and $m$ edges, there is a randomized algorithm that runs in depth $O\left(\log^2n\right)$ and work $O\left(m\cdot n^{3.5}\log^2n\right)$ and computes a maximum matching of $G$ with probability at least $1/2$. 
\end{theorem}
\noindent
We will use the following standard random walk fact for our \efk{\left(k\right)} algorithm for $k  \approx \sqrt{m}$.
\begin{restatable}[Random Walk Fact]{lemma}{hoeffding}\label{lemma:hoeffding}
    An $N$-step random walk on $\mathbb{Z}$ beginning at $0$ that takes a $+1$ step with probability $p\geq 1/2$ and a $-1$ step with probability $1-p$ goes below $-\sqrt{cN\log (dN)}$ with probability at most $(dN)^{-c}$ for any $c,d\geq1$.
\end{restatable}
\begin{proof}
    Let $X_t\in\{-1,1\}$ be the value of the step of the random walk at time $t$ and let $X:=\sum_{t\in[N]}X_t$. It suffices to prove that $|X|>\sqrt{cN\log(dN)}$ with probability at most $(dN)^{-c}$. We have
    \begin{align*}
        \operatorname{Pr}\left(|X|> \sqrt{cN\log(dN)}\right)
        &\leq\operatorname{Pr}\left(|X-\mathbb{E}[X]|> \sqrt{cN\log(dN)}\right)  
        \\&\leq \exp\left(-\sqrt{\frac{cN\log (dN)}{N}}^2\right)  
        \\ &= \left(\frac{1}{dN}\right)^{-c}
    \end{align*}
    where the first line is because $\mathbb{E}[X]=\sum_t\mathbb{E}[X_t]\geq\sum_t 0\geq0$ and the second line is by Hoeffding's inequality. Then with probability at most $(dN)^{-c}$ the random walk goes below $-\sqrt{cN\log(dN)}$. 
\end{proof}

\paragraph*{Fair Allocation Background.}
Having defined our model for parallel computation, we now formally define the instances and relevant notations for the allocation problem. 
\begin{definition}[Allocation Instances with Additive Valuations]\label{defn:instances}
    An instance of the allocation problem is the tuple $\mathcal{I}=\left(M,N,\{v_i\}_{i\in[N]}\right)$ where $M$ is the set of $m$ goods, $N$ is the set of $n$ agents, and $v_i:2^M\to\mathbb{R}_{\geq0}$ is agent $i$'s \defn{additive} value function over goods, i.e.\ for any $S\subseteq M$ and any $i$ we have $v_i(S)=\sum_{j\in S}v_i(j)$. Given $\mathcal{I}$, an \defn{allocation} $\mathcal{A}$ is a partition of the goods into $n$ parts where each part is assigned to a unique agent, and $\cA(i)$ denotes the part assigned to agent $i$. A \defn{partial allocation} $\cA$ is a partition of a subset of the goods where each part is assigned to a unique agent.
\end{definition}
\noindent For the rest of the paper, we assume additive valuations for the agents. We will assume without loss of generality that $m/n$ is an integer by adding goods that no agent positively values. Note this only increases the number of goods by at most a multiplicative factor of $2$, so any bounds in terms of $m$ are only affected by a factor of $2$. Given an instance of the allocation problem, we say that an agent's \textit{\textbf{preference list}} is an ordering of the goods where the first good is that agent's highest valued good and so on, with ties broken arbitrarily. An agent's preference list is \textbf{\textit{complete}} if that agent has a positive value for all goods, and \textbf{\textit{incomplete}} otherwise.

The following characterization of the given allocation instances will also be useful, which we study in \Cref{subsec:graph,subsec:3hypergraph}.
\begin{definition}[Induced Hypergraph]\label{defn:graphhypergraph}
    Given an allocation instance $\mathcal{I}$, we say $\cI$ induces the hypergraph $G$ which is defined as follows: make a vertex for each agent, and for each good $g$ that is positively valued by a subset of agents $S\subseteq N$, we add a hyperedge $g$ containing the vertices $S$.
\end{definition}
\noindent Hence the agent and good sets $N,M$ are interchangeable with the vertex and edge sets $V(G),E(G)$ respectively. When convenient we will directly refer to an allocation instance as its induced hypergraph. We next define the fairness guarantee that we'll be studying:
\begin{definition}[Envy-Free Up to One Good (\efone)]\label{defn:ef1}
    An allocation $\cA$ is envy-free up to one good \textsf{(EF1)} iff for every pair of agents $i,j$ there exists a good $s\in \cA(j)$ such that $v_i(\cA(i))\geq v_i(\cA(j)\setminus \{s\})$.
 \end{definition}
\noindent Next we formally define the Round Robin algorithm that achieves \efone.
\begin{definition}[Round Robin]
    The Round Robin procedure takes as input an instance $\mathcal{I}$ and set of permutations of the agents $\{\sigma_i\}_{i\in[ m/n]}$. For each round $i\in[ m/n]$, agents choose their highest-valued remaining good in the order given by $\sigma_i$. Fixed-Order Round Robin is the Round Robin procedure where all the $\sigma_i$'s are identical. 
\end{definition}
\noindent 
\noindent We next define the envy-cycle elimination algorithm of \cite{lipton04approximately}, which also achieves \efone.
\begin{definition}[Envy Graph and Envy-Cycle Elimination]\label{defn:envycycle}
Given a partial allocation $\mathcal{A}=(A_1,\ldots,A_n)$, the \defn{envy graph} is the directed graph on the agents, with a directed edge $(i,j)$ if agent $i$ envies agent $j$, i.e., if $v_i(A_i) < v_i(A_j)$. A vertex of indegree $0$ in the envy graph is called a \defn{source}.

The \defn{envy-cycle elimination} procedure starts from the empty allocation and maintains the envy graph of the current partial allocation. At each step, it starts with an acyclic envy graph, selects an arbitrary unallocated good $g$, and assigns it to a source vertex of the current envy graph. If this creates a directed cycle $(i_1,i_2,\ldots,i_k)$, then the procedure \defn{eliminates} the cycle by cyclically permuting the bundles along the cycle; that is, agent $i_a$ receives the bundle previously held by agent $i_{a+1}$ for each $a\in[k-1]$, and agent $i_k$ receives the bundle previously held by agent $i_1$. The envy graph is then updated, and this cycle-elimination step is repeated until the envy graph becomes acyclic again.
\end{definition}
\begin{restatable}[\cite{lipton04approximately}]{lemma}{envycycle}
\label{lemma:envycycle}
    Given a partial allocation $\cA$ that is \efone, repeated application of envy-cycle elimination on its corresponding envy graph produces a partial allocation $\cA'$ whose envy graph is acyclic and such that $\cA'$ is still \efone.
\end{restatable}
\noindent We will also study the following weaker fairness notion.
\begin{definition}[Envy-Free Up to $k$ Goods (\efk{\left(k\right)})]
    An allocation $\cA$ is \efk{\left(k\right)} iff for every pair of agents $i,j$ there exists $k$ goods $S:=\{s_1,s_2,\dots,s_k\}\subseteq \cA(j)$ such that $v_i(\cA(i))\geq v_i(\cA(j)\setminus S)$.
\end{definition}

\section{A Faster NC Algorithm for \efone with Two Agents}

In this section, we give an $O(\log m)$ depth and $O(m)$ work algorithm for \efone allocations with $2$ agents, assuming additive utilities and apply this to graph instances. As earlier mentioned, our algorithm is based on a parallel implementation of a certain recursive game.

More formally, we are given an allocation instance $\mathcal{I}=(M,N,\{v_i\}_i)$ with $2$ agents $N=\{a_1,a_2\}$ and $m$ a power of $2$ (which we may assume without loss of generality by padding with  goods to which both agents give value $0$). When convenient, we will also refer to these agents as $1$ and $2$. 
We label the goods $M$ arbitrarily as $\left\{g^0_i\right\}_{i\in[m]}$, and recursively construct a binary tree $T$ bottom-up as follows: at level $0$ are the leaves, which are $\left\{g^0_i\right\}_{i\in[m]}$. Let $m_0=m$ and $m_{l+1}= m_l/2$. At level $l\geq 1$ of the tree we have nodes $\left\{g^l_i\right\}_{i\in[m_l]}$. For $i\in[m_l]$, node $g^l_i$ has exactly $2$ children $g^{l-1}_{2i-1},g^{l-1}_{2i}$. It is clear this construction produces a balanced binary tree. 

The agents $a_1,a_2$ then play the following game on $T$, starting at the root $r$. At the root $r$, $a_1$ gets to choose first one of the two subtrees $L(r),R(r)$ from the root based on a rule described below; without loss of generality suppose $a_1$ chooses $L(r)$, so $a_2$ must choose $R(r)$. We then repeat for $r$'s children, where $a_1$ chooses first in the game associated with $L(r)$, and $a_2$ chooses first in the game associated with $R(r)$. This continues recursively down the tree, until $a_1,a_2$ have chosen all of the leaves. We obtain an allocation $\cA$ from the leaves that each agent chooses in the game by allocating to agent $i$ the goods corresponding to the leaves it chose in the game. 

This game admits the following alternate description which we will later use to prove the resulting allocation is \efone. 
First, at the root $a_1$ chooses a root--leaf path and claims the corresponding good. This root--leaf path corresponds to a sequence of first moves made by $a_1$.
Removing this root--leaf path disconnects the tree into $\log m$ disjoint subtrees, each with $1, 2, 4, \ldots, m/2$ leaves. 
Agent $a_2$ then moves first on \emph{all} of these disjoint subtrees, choosing root--leaf paths, claiming $\log m$ goods, and producing still more subtrees with between $1$ and $m/4$ leaves. Agent $a_1$ then moves first on these subtrees, and this continues until all goods are claimed.
This is illustrated in \Cref{fig:binary}. 

\begin{figure}[h]
    \centering
    \begin{subfigure}[t]{0.32\textwidth}
        \centering
        \includegraphics[width=.95\linewidth]{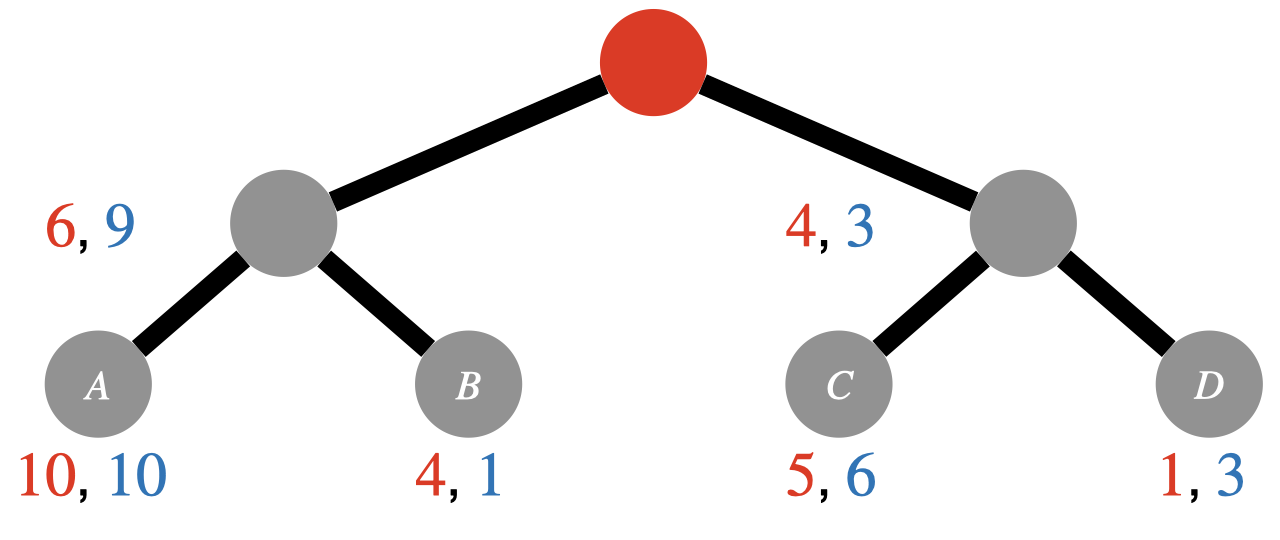}
        \subcaption{Binary tree $T$.}
        \label{subfig:binary1}
    \end{subfigure}%
    ~
    \begin{subfigure}[t]{0.32\textwidth}
        \centering
        \includegraphics[width=.95\linewidth]{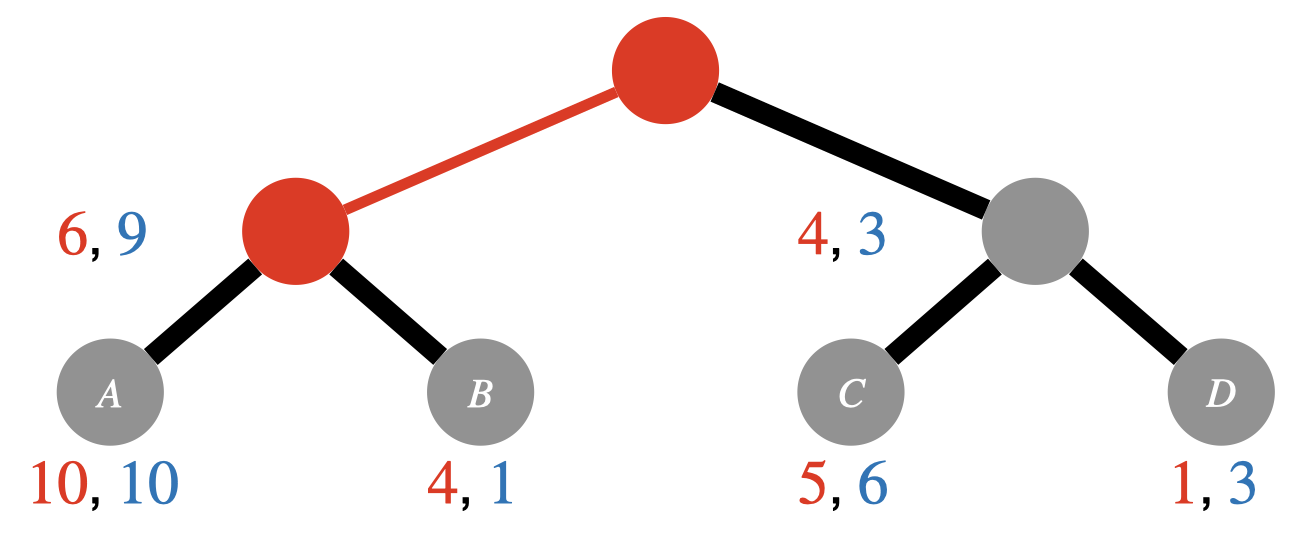}
        \subcaption{Red chooses the left subtree.}
        \label{subfig:binary2}
    \end{subfigure}%
    ~
    \begin{subfigure}[t]{0.32\textwidth}
        \centering
        \includegraphics[width=.95\linewidth]{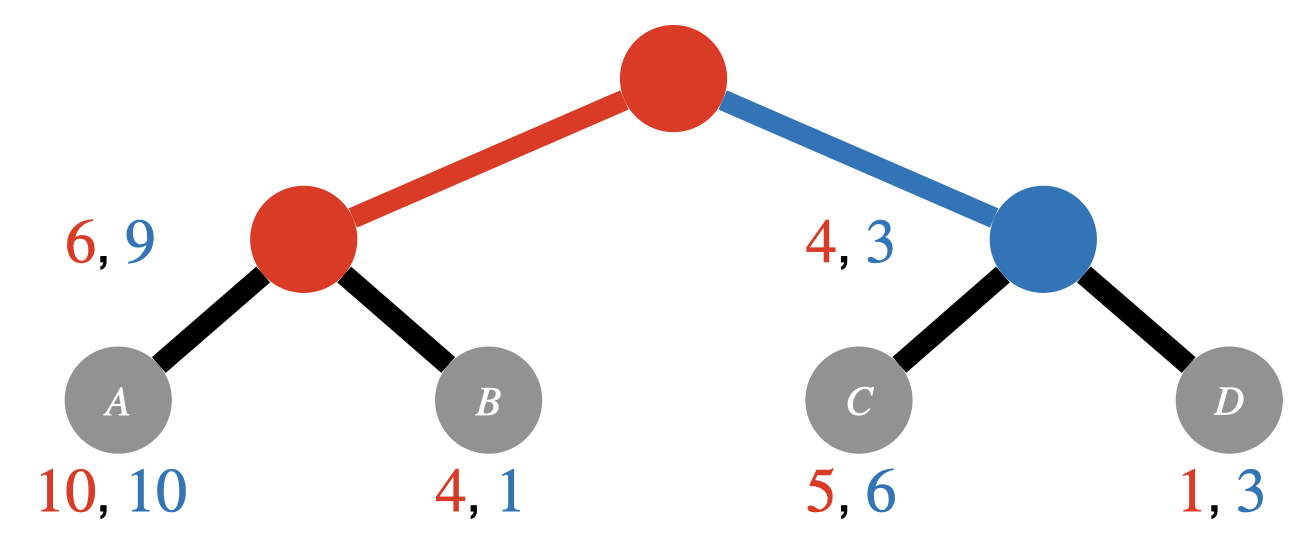}
        \subcaption{Blue chooses the right subtree.}
        \label{subfig:binary3}
    \end{subfigure}%
    
    \begin{subfigure}[t]{0.32\textwidth}
        \centering
        \includegraphics[width=.95\linewidth]{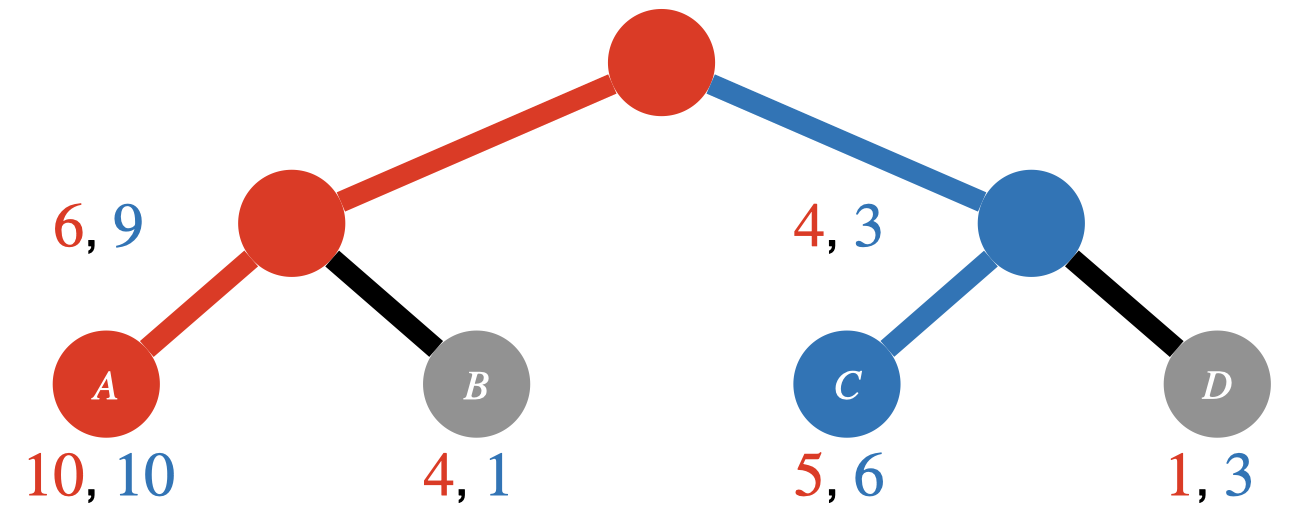}
        \subcaption{The next level of first choices.}
        \label{subfig:binary4}
    \end{subfigure}%
    ~
    \begin{subfigure}[t]{0.32\textwidth}
        \centering
        \includegraphics[width=.95\linewidth]{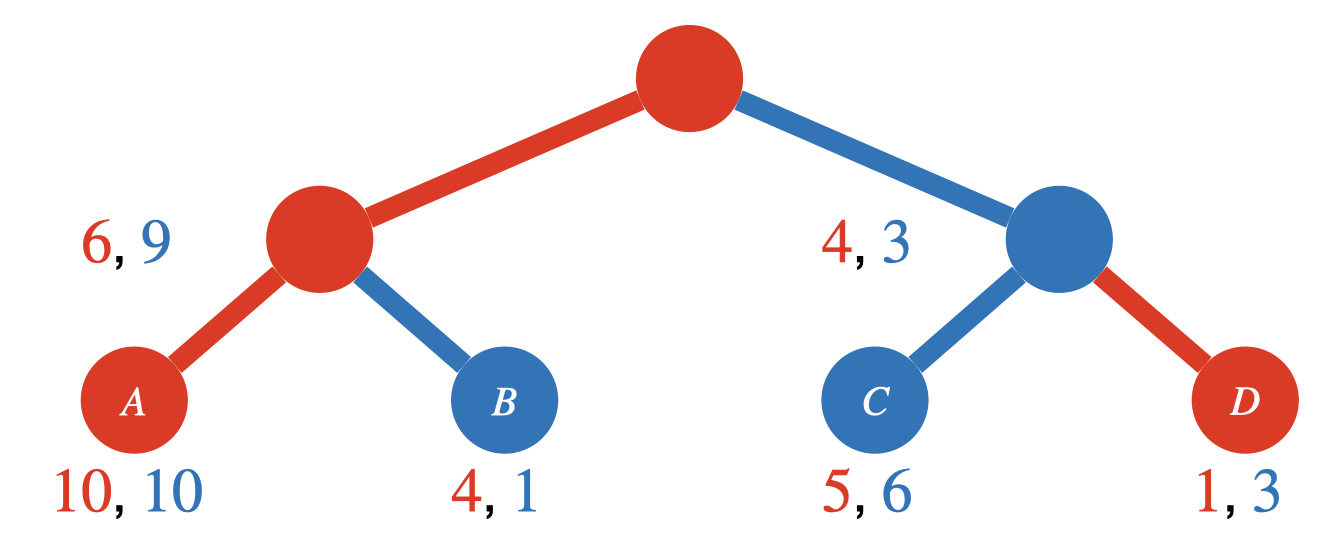}
        \subcaption{The next level of second choices.}
        \label{subfig:binary5}
    \end{subfigure}%
    ~
    \begin{subfigure}[t]{0.32\textwidth}
        \centering
        \includegraphics[width=.95\linewidth]{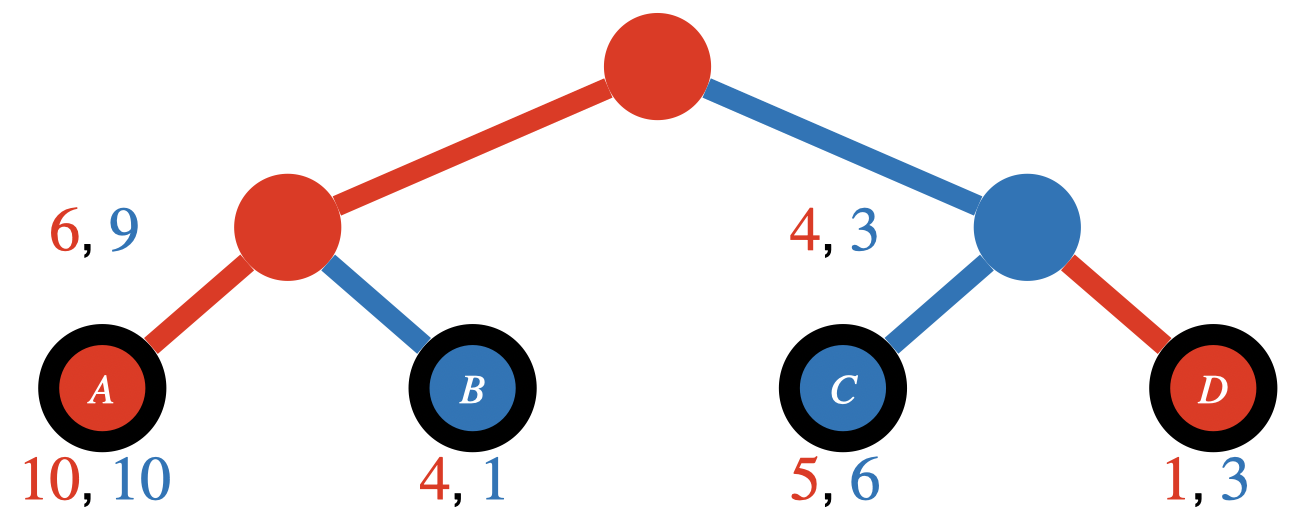}
        \subcaption{Allocation $\cA$ is given by the leaves.}
        \label{subfig:binary6}
    \end{subfigure}%
    \caption{The algorithm for two agents red and blue, where we arbitrarily let red choose first at the root. We construct a balanced binary tree $T$ whose leaves correspond to the goods, labeled $A,B,C,D$. The two agents' values of the goods and utility gaps are given in red and blue by the corresponding nodes. Each node is highlighted with the color corresponding to the agent that chooses first at it, and each edge is highlighted with the color corresponding to the agent that chose the subtree descending it. 
    } 
    \label{fig:binary}
\end{figure}
\noindent
We now describe the agent strategies that define the game. 
For $i \in \{1,2\}$, let $\neg i$ denote the agent in $\{1,2\}$ that is not $i$.
At each node $u$ the agent who moves first will choose the subtree that maximizes the following function on $T$.
\begin{definition}[Utility Gap]\label{defn:gap}
    Let $\cA_i(u,j)$ denote the allocation that agent $i$ receives if $j$ moves first at node $u$; then we denote the \defn{utility gap} for $i$ in subgame $(u,j)$ by
    \[
        \gap_i^j(u) := v_{i}(\cA_i(u,j))- v_i(\cA_{\neg i}(u,j)).
    \]
\end{definition}
\noindent In other words, $\gap_i^j(u)$ denotes how much more agent $i$ prefers their allocation to that of the other agent, in the allocation obtained if agent $j$ moves first at node $u$. (Note that $i$ and $j$ can be the same agent.) To recap: our algorithm outputs the allocation $\cA$ that arises when each agent $i$ who moves first in each subgame $u$ chooses the child in $\{L(u),R(u)\}$ that maximizes $\gap_i^i(\cdot)$, breaking ties deterministically.

When both agents play this strategy, the utility gap exhibits recursive structure; namely, the gap of one node is the sum of the gaps of its children.
\begin{observation}
\label{obs:recursive-gap-formula}
    Let $C_j(u)\in \{L(u),R(u)\}$ be the child that agent $j$ chooses when moving first at $u$ in this game, and let $D_j(u)$ be the other child.
    Then for additive utilities, the utility gap obeys
    \[
        \gap_i^j(u) = \gap_i^j(C_j(u)) + \gap_i^{\neg j}(D_j(u)).
    \]
\end{observation}
\begin{proof}
    If $j$ moves first, then
    \[
        \cA_i(u,j) = \cA_i(C_j(u),j) + \cA_i(D_j(u),\neg j).
    \]
    Then by utility additivity,
    \begin{align}
        \gap_i^j(u) &= v_i(\cA_i(u,j)) - v_i(\cA_{\neg i}(u,j)) \notag \\
        &= v_i(\cA_i(C_j(u),j) + \cA_i(D_j(u),\neg j)) - v_i(\cA_{\neg i}(C_j(u),j) + \cA_{\neg i}(D_j(u),\neg j)) \notag  \\
        &= \gap_i^j(C_j(u)) + \gap_i^{\neg j}(D_j(u)). \qedhere
    \end{align}
\end{proof}
\noindent 
To analyze the allocation $\cA$ that results from this game, we will leverage other important properties of the utility gap $\gap$. In particular, we use the fact that a agent can only achieve positive utility by going first and the amount of gap they gain by going first is at least the amount of gap they lose by going second.
\begin{lemma}\label{lemma:gap}
    For any node $u\in V(T)$ and agent $i \in \{1,2\}$, we have that $\gap_i^i(u) \geq 0$ and $\gap_i^i(u) \geq -\gap_i^{\neg i}(u)$.
\end{lemma}
\begin{proof}  
    We show this by induction on the number of levels of the binary tree. On a $1$-level tree (i.e.\ a leaf node $g$), 
    suppose that $i$ chooses first; then $i$ chooses the good and $\neg i$ gets nothing; hence
    \[
        \gap_i^i(g) = v_i(\cA_i(g,i)) - v_i(\cA_{\neg i}(g, i)) = v_i(g) - v_i(\emptyset) = v_i(g) \geq 0.
    \]
    If $\neg i$ moves first then they take the good, and so
    $\gap_i^{\neg i}(g) = v_i(\emptyset) - v_i(g) = - v_i(g)$. Hence the second claim holds on 1-level trees with equality. 

    For the inductive step, suppose $u$ is an internal node with children $L$ and $R$, and the inductive hypotheses apply for all $i\in \{1,2\}$ on $L$ and $R$.
    We first address gap nonnegativity. 
    When both agents play the prescribed strategy, \Cref{obs:recursive-gap-formula} implies
    \begin{align*}
        \gap_i^i(u) &= \gap_i^i(C_i(u)) + \gap_i^{\neg i}(D_i(u)) \\
        &\geq \gap_i^i(D_i(u)) + \gap_i^{\neg i}(D_i(u)) \\
        &\geq 0,
    \end{align*}
    where the last step follows from the inductive hypothesis (second claim) on $i$ and $D_i(u)$.

    We now show inductive step for the second claim.
    First, suppose there is no contention: $i$ and $\neg i$ choose different subtrees, and so $C_i(u) = D_{\neg i}(u)$. Then
    \[
        \cA_i(u,i) = \cA_i(C_i(u),i) + \cA_i(D_i(u), \neg i) = \cA_i(D_j(u),i) + \cA_i(C_{\neg i}(u),\neg i) = \cA_i(u,\neg i),
    \]
    (and similarly for $\cA_{\neg i}(u, i) = \cA_{\neg i}(u,\neg i)$), and so $\gap_i^i(u) = \gap_i^{\neg i}(u)$ by the definition of the utility gap.
    Since $\gap_i^i(u) \geq 0$, the claim holds.

    For the remaining case, suppose $i$ and $\neg i$ share a preferred subtree: $C_i(u) = C_{\neg i}(u)$. 
    Then 
    \begin{align*}
        \gap_i^i(u) &= \gap_i^i(C_i(u))+\gap_i^{\neg i}(D_i(u)) 
        && \text{by \Cref{obs:recursive-gap-formula} }  \\
        &\geq -\gap_i^{\neg i}(C_i(u)) + \gap_i^{\neg i}(D_i(u))
        && \text{by hypothesis on $C_i(u)$} \\
        &\geq -\gap_i^{\neg i}(C_i(u)) - \gap_i^i(D_i(u)) 
        && \text{by hypothesis on $D_i(u)$} \\
        &= - \left(\gap_i^{\neg i}(C_{\neg i}(u)) + \gap_i^i(D_{\neg i}(u)) \right)
        && \text{by shared preferred subtree} \\
        &= -\gap_i^{\neg i}(u) 
        && \text{by \Cref{obs:recursive-gap-formula}},
    \end{align*}
    where for the inequalities we applied the inductive hypothesis (second claim) first on $C_i(u)$, then on $D_i(u)$.
    This concludes the proof.
    %
    %
    %
    %
\end{proof}
\noindent We are ready to prove that the algorithm obtains an \efone allocation.
\twoagents*
\begin{proof}
    We will show that the allocation found by the tree game is \efone for two agents with additive utilities, and moreover that its outcome can be computed in $O(m)$ work and $O(\log m)$ depth.

    Let $i$ be the agent who plays first at the root $r$. 
    By \Cref{lemma:gap} applied to the root node, we have 
    \[
        \gap_i^i(r) = v_{i}(\cA_i(r,i))- v_i(\cA_{\neg i}(r,i)) \geq 0.
    \]
    Since $\cA_{\neg i}(r, i) = M \setminus \cA_i(r,i)$ is precisely the other agent's final allocation according to $A$, this means that $i$ prefers $i$'s allocation to their opponent's.

    To establish \efone, we must therefore show that $\neg i$ envies $i$ up to at most one good. 
    To see this, consider the root-leaf path of first choices that $i$ makes, which leads to a single good. 
    Removing this good $g^*$ and corresponding path splits the game into $\log m$ disjoint subgames $r_1, \ldots, r_{\ell}$, where $\neg i$ plays first in all of them. 
    Letting $M(u)$ be the goods in the subtree of $u$, by the argument above we have $v_{\neg i}(\cA_{\neg i}(r_k,\neg i))- v_{\neg i}(\cA_{i}(r_k,\neg i)) \geq 0$ for all $k \in [\ell]$; summing over the disjoint subgames and using the additivity of utility,
    \begin{align*}
        v_{\neg i}(\cA_{\neg i}(r,i)) = \sum_k v_{\neg i}(\cA_{\neg i}(r_k,\neg i)) 
        \geq \sum_k v_{\neg i}(\cA_{i}(r_k,\neg i)) = v_{\neg i}(\cA_{i}(r, i) \setminus \{g^*\}).
    \end{align*}
    Since $\neg i$ does not envy $i$'s allocation when $g^*$ is withheld, $\neg i$ envies $i$ up to at most one good in the final allocation.

    Finally, the outcome of the game can be computed in depth $O(\log m)$ and work $O(m)$ by computing for each node $u\in V(T)$ the values $\gap^1_1(u),\gap_1^2(u),\gap_2^1(u),\gap_2^2(u)$ via a bottom-up dynamic program over $T$, where $T$ is a tree with depth $\log m$ and $|V(T)|=O(m)$.
\end{proof}
\subsection{An \texorpdfstring{$O(\log m)$}{O(logm)}-Depth Algorithm for Graph Instances}
\label{subsec:graph}
As an application of our faster two-agent algorithm, we give an equally fast \NC algorithm for instances in which each good is positively valued by at most two agents. By \Cref{defn:graphhypergraph}, these instances induce (non-simple) graphs, where the vertices correspond to agents and edges correspond to goods. In particular, there may be parallel edges and self-loops. Also note that for these instances we do not constrain the number of agents $n$, and it can be arbitrarily large.
While an \NC algorithm for two agents was previously known, it was not observed that this could be used to obtain \NC algorithms for graph-structured instances. 
\begin{corollary}[Graphs]
    There is an algorithm for allocation instances where each good is positively valued by at most $2$ agents that outputs an \efone allocation with depth $O(\log m)$ and work $O(m)$.
\end{corollary}
\begin{proof}
    We may assume that there are no goods that are positively valued by $0$ agents, and we observe that goods that are positively valued by exactly $1$ agent correspond to self-loops in $G$. 
    
    For a vertex $i$, let $E_i$ be the set of self-loops on $i$, and for a pair of vertices $i,j$, let $E_{i,j}$ be the set of parallel edges between them. Then for each vertex $i$, the algorithm allocates $E_i$ to agent $i$ obtaining an allocation $\cA_i$, and for each unordered pair of vertices $i,j$, the algorithm runs the algorithm for two agents (\Cref{thm:2-agents}) in parallel on the instance $\cI_{i,j}:=\{E_{i,j}, \{i,j\},\{v_i,v_j\}\}$ to get an allocation $\cA_{i,j}$ of the goods $E_{i,j}$ to the agents $i,j$. The algorithm returns $\cA:=\bigcup_i\cA_i \cup \bigcup_{i,j}\cA_{i,j}$. Allocating the self-loops can be done in constant time, so by \Cref{thm:2-agents}, the algorithm in total has depth $O(\log m)$ and work $\sum_{i,j}O(|E_{i,j}|)\leq O(m)$.

    It remains to show that the algorithm outputs an \efone allocation. Fix any pair of agents $i,j$. Agent $i$ has $0$ value for any goods not in $E_{i,j}$ that were allocated to agent $j$ and vice-versa, since each good is valued by at most $2$ agents in this instance. Then by \Cref{thm:2-agents}, we have that there is a good $g\in \cA(j)\cap E_{i,j}$ such that $ v_i(\cA(i))\geq v_i(\cA(i)\cap E_{i,j})\geq v_i(\cA(j)\cap E_{i,j}\setminus\{g\}) = v_i(\cA(j)\setminus\{g\})$, and vice-versa. 
\end{proof}

\section{An \NC Algorithm for \texorpdfstring{$O(1)$}{O(1)} Agents}\label{subsec:constantagents}

In this section, we show that computing an \efone allocation when $n=O(1)$ is in \NC, improving upon the algorithm of \cite{garg2023fairly} which handled only $n\leq 3$. We then apply this result to obtain \NC algorithms for sparse hypergraph instances and, in particular, \NC algorithms when each agent values at most $\polylog(m)$ goods and each good is valued by $O(1)$ agents. 

More specifically, rather than prove \Cref{thm:reachabilityreduction} (our \NC algorithms for $O(1)$ agents), we prove the more general result below.
\begin{restatable}{theorem}{reachabilityreductionGen}\label{thm:reachabilityreductionGen}
    There is an $O(n^2\log^2 m)$ depth and $O(nm^{n+1})$ work algorithm that, given agents with additive valuations and a fixed ordering on these agents, outputs the \efone allocation of Round Robin with this ordering.
\end{restatable}

\noindent It is clear that \Cref{thm:reachabilityreduction} follows immediately from \Cref{thm:reachabilityreductionGen} so for the rest of the section we focus on proving \Cref{thm:reachabilityreductionGen}. We will then apply \Cref{thm:reachabilityreductionGen} to obtain the aforementioned sparse hypergraph results.

At a high level, we view Fixed-Order Round Robin as a deterministic process on a polynomial-sized state space of partial allocations. We show that, given a configuration describing the current positions in the agents' preference lists, the next configuration after one round of Round Robin with a fixed agent ordering on the residual instance can be computed in $O(\log m)$ depth and $O(m)$ work. 
We will precompute in parallel the successor of every valid configuration, of which there are $O(m^n)$.
This defines a directed graph, which we refer to as the \emph{reachability graph}, in which every configuration has out-degree 1, and the configurations reachable from the source form a directed path corresponding exactly to the Round Robin execution. We then recover this path in parallel using parallel tree contraction.

\paragraph*{Notation.}
Let $\pi_i = (\pi_i(1), \pi_i(2), \ldots, \pi_i(m))$ denote the preference list of agent $i$, where $\pi_i(j)$ is the $j$th preferred good of agent $i$. We define a \defn{configuration} to be a tuple $C=(c_1,c_2,\ldots,c_n)$ of indices with $c_i\in [m+1]$ for each $i\in [n]$, which we interpret as a partial allocation state, or residual instance, in which:
\begin{itemize}
    \item Agent $i$ picks goods according to its \defn{residual preference list} $(\pi_i(c_i), \pi_i(c_i+1),\ldots, \pi_i(m))$, i.e., the suffix of $\pi_i$ starting at position $c_i$.
    \item Every good appearing in the prefix $(\pi_i(1), \pi_i(2), \ldots, \pi_i(c_i-1))$ for any agent $i$ is assumed to be already allocated. We denote this set of pre-allocated goods by
    \begin{align*}
        \cP(C):= \bigcup_{i=1}^n \left\{\pi_i(1),\pi_i(2),\ldots, \pi_i(c_{i}-1)\right\}.
    \end{align*}
\end{itemize}
Observe that the configuration $(1,1,\ldots,1)$ corresponds to the initial state (no goods allocated).

\paragraph*{Simulating a Single Round of Round Robin.}
To construct the edges of the reachability graph, we must determine, for each configuration $C=(c_1,c_2,\ldots,c_n)$, the configuration obtained after one round of Fixed-Order Round Robin on the corresponding residual instance. Given such a configuration $C$, let $C'=(c'_1,\ldots,c'_n)$ denote its successor configuration, where, if agent $i$ is assigned the good $\pi_i(t_i)$ during this round, then $c'_i=t_i+1$. Note that if any agent is not assigned a good, then the configuration $C$ has no successor. The next lemma shows that this successor configuration can be computed efficiently in parallel.

\begin{lemma}\label{lemma:single_round_residual}
Given a configuration $C$, the configuration $C'$ obtained by simulating one round of Fixed-Order Round Robin with ordering $\sigma$ on the agents on the residual instance can be computed in $O(\log m + n^2)$ depth and $O(nm + n^2)$ work.
\end{lemma}
\begin{proof}
    In order to simulate a single round of Round Robin, we first update the residual preference lists of the agents by removing all pre-allocated goods corresponding to configuration $C$. We begin by computing the set $\cP(C)$. Then, for each agent, we filter these goods out from its preference list using a parallel filter. By \Cref{lemma:parallel_filter}, this takes $O(\log m)$ depth and $O(m)$ work per agent, for a total of $O(\log m)$ depth  and $O(nm)$ work overall. 

    Given the filtered preference lists, we then simulate one round of Round Robin sequentially according to $\sigma$. Observe that, in a filtered list, an agent only needs to skip over goods picked by agents preceding it in the current round. Since there are at most $n-1$ such goods, each agent advances by at most $n$ positions in its filtered list. It follows that this sequential simulation takes $O(n^2)$ depth and work. During this step, we either determine the good assigned to each agent, and hence the successor configuration $C'$, or detect that some agent is not assigned any good, in which case no successor exists.

    Overall, the total depth is $O(\log m+n^2)$ and the total work is $O(nm+n^2)$.
\end{proof}

\paragraph*{Reachability Graph.}
We construct a directed graph $G=(V,E)$ whose vertices encode configurations and whose edges encode valid Round Robin transitions, as follows:

\begin{itemize}
    \item The vertex set $V$ contains one node for every possible configuration $C \in [m+1]^n$. The source node corresponds to the configuration $s:=(1,1,\ldots,1)$. The total number of vertices is $|V|=(m+1)^n = O(m^n)$.
    \item Each node $u \in V$, corresponding to a configuration $C$, stores the partial allocation produced by simulating one round of Fixed-Order Round Robin with fixed agent order $\sigma$ on the residual instance defined by $C$.
    \item For any two nodes $u,v \in V$, there is a directed edge from $u$ to $v$ if and only if the Round Robin simulation at $u$ allocates a good to every agent, and $v$ is the resulting successor configuration. If some agent is not assigned a good in that round, then we do not add any out-edge from $u$.
\end{itemize}
Since Fixed-Order Round Robin is deterministic, every node has outdegree at most one. Moreover, every transition allocates at least one new good, and hence strictly increases the number of allocated goods. It follows that the subgraph reachable from $s$ is a directed path. One can recover this path efficiently using parallel tree contraction as follows.

\begin{lemma}\label{lem:tree_contraction_path}
Let $G=(V,E)$ be the reachability graph, and let $s$ be the source configuration. Then the unique directed path in $G$ obtained by repeatedly following outgoing edges from $s$ can be computed in $O(\log^2 |V|)$ depth and $O(|V|)$ work.
\end{lemma}
\begin{proof}
Since every node of $G$ has out-degree at most one, and every edge strictly increases the number of allocated goods, $G$ forms a directed forest. By adding directed edges from all sinks to a new super-sink node $t$, we obtain a rooted directed tree. Thus, it suffices to recover the unique directed path from $s$ to $t$.
We recover this path using parallel tree contraction (\Cref{lemma:treecontraction}). Initially, we mark only the vertex $s$. During each contraction step, whenever several vertices are merged into a single cluster, we mark the new cluster if and only if it contains a marked vertex. Thus, throughout the contraction phase, the unique cluster containing $s$ remains marked. Once the tree has been contracted to a single root cluster, we perform the standard uncontraction phase, propagating the mark back only through those clusters that were marked during contraction. It follows that, after uncontraction, exactly the vertices on the $s$--$t$ path are marked, and hence exactly the vertices on the directed path from $s$ to the sink in $G$ are marked.

We can then apply a parallel filter to the vertices of $G$ to extract the marked nodes. By \Cref{lemma:treecontraction}, the tree contraction step takes $O(\log^2 |V|)$ depth and $O(|V|)$ work, while the final filtering requires $O(\log |V|)$ depth and $O(|V|)$ work. The stated bounds follow.
\end{proof}


\noindent We are now ready to prove that the algorithm obtains a Fixed-Order Round Robin allocation, and thus, an \efone allocation.

\reachabilityreductionGen*
\begin{proof}
    We first argue correctness. Starting from the source configuration $s=(1,\ldots,1)$, consider the unique path obtained by repeatedly following outgoing edges. By construction, each node on this path corresponds to the configuration at the beginning of a round of Fixed-Order Round Robin, and stores exactly the partial allocation produced in that round. Thus, by induction on the number of rounds, the $r$th node on this path is exactly the configuration reached after $r$ rounds of Round Robin. Since every transition allocates at least one previously unallocated good, this path must terminate. Therefore, the unique directed path starting at $s$ in $G$ exactly corresponds to the execution of Fixed-Order Round Robin, and thus the corresponding allocation is \efone.

    We now analyze the complexity. By \Cref{lemma:single_round_residual}, the outgoing edge for each configuration, together with its associated partial allocation, can be computed in $O(n^2+\log m)$ depth and $O(nm+n^2)$ work. Computing these independently in parallel for all $O(m^n)$ configurations takes $O(n^2+ n\log m)$ depth and $O(nm^{n+1})$ total work. By \Cref{lem:tree_contraction_path}, the required path and corresponding allocation can be computed in $O(n^2\log^2 m)$ depth and $O(m^n)$ work. Thus, overall, the algorithm runs in $O(n^2\log^2 m)$ depth and $O(nm^{n+1})$ work.
\end{proof}

\subsection{An \NC Algorithm for Sparse \texorpdfstring{$O(1)$}{O(1)}-Rank Hypergraph Instances}
\label{subsec:3hypergraph}
In this section, we compute \efone allocations for instances inducing hypergraphs of rank $r=O(1)$ and bounded edge degree $\Delta = \polylog(n,m)$.
The high-level approach is similar to the algorithm on graph instances in \Cref{subsec:graph}.
The main new difficulty is that, unlike in the graph case, two hyperedges may intersect in more than one agent. As a result, the envy relations created by one hyperedge need not be compatible with the other, and thus, we cannot process all hyperedges independently.

To handle this, we consider the graph whose vertices are the distinct hyperedges, with two hyperedges adjacent if and only if they intersect (i.e., the line graph of the input hypergraph). Since this graph has maximum degree at most $\Delta$, it admits a proper $(\Delta+1)$-coloring. We process the color classes sequentially. Hyperedges of the same color are pairwise disjoint, and hence can be handled independently in parallel.

For a distinct hyperedge $e$, let $M_e$ denote the set of goods inducing $e$. Before processing a color class, we determine, for each hyperedge $e$ in that class, a suitable ordering of the agents in $e$ from the current envy graph. Concretely, we consider the envy graph induced by the agents of $e$ under the current partial allocation. If this induced envy graph contains a cycle, we repeatedly apply the standard sequential envy-cycle elimination algorithm (\Cref{defn:envycycle}) restricted to the agents of $e$, until the induced envy graph becomes acyclic. Since $|e|\le r=O(1)$, this takes $O(1)$ depth and work. We then fix any topological ordering of the resulting DAG, and use this ordering to allocate the goods in $M_e$ via the constant-agent algorithm from the previous subsection. Note that, unlike the previous subsection, the ordering is no longer fixed globally, and may change from one color class to the next.

We now formalize the above discussion. Observe that, below, when $\Delta=\polylog(m,n)$, the algorithm is an \NC algorithm.

\begin{theorem}\label{thm:hypergraph_nc}
There is an algorithm for hypergraph instances of rank $r=O(1)$ and edge degree at most $\Delta$, that computes an \efone allocation in $O(\Delta\log^4 m)$ depth and $O(m^{r+1})$ work.
\end{theorem}

\begin{proof}
Let $E$ denote the set of distinct hyperedges, and for each $e\in E$ let $M_e$ denote the set of goods inducing $e$. Compute a proper $(\Delta+1)$-coloring of the hyperedges in $E$, and process the resulting color classes one after another. For each color class, all hyperedges in the class are pairwise disjoint, and processing a color class means allocating the goods in the hyperedges in this color class.

We prove correctness by induction on the number of processed color classes. At the outset, before any color class is processed the allocation is empty, and hence trivially \efone. Now suppose that after processing the first $c-1$ color classes, the current allocation is \efone. 
Let this allocation be denoted by $A$, and consider the $c$th color class.

Fix a hyperedge $e$ in this color class. First, we update the allocation $A$ by repeatedly eliminating cycles in the envy graph induced by the agents of $e$, restricted to these agents alone, until the induced envy graph becomes acyclic. Since envy-cycle elimination preserves \efone (\Cref{lemma:envycycle}), the resulting allocation $A$ remains \efone. 
We now fix any topological ordering $\sigma_e$ of the resulting induced envy DAG on $e$, and let $B$ denote the allocation of the goods in $M_e$ produced by the constant-agent algorithm using the ordering $\sigma_e$.

Consider any two agents $i,j\in e$, and suppose that $i$ precedes $j$ in $\sigma_e$. Since $\sigma_e$ is a topological ordering, there is no edge from $j$ to $i$ in the induced envy graph under $A$, and hence $j$ does not envy $i$ under $A$.
Since $A$ is \efone, agent $i$ envies $j$ by at most one good under $A$. On the other hand, by the guarantee of the constant-agent algorithm, the allocation $B$ is \efone on the agents of $e$, and moreover, since $i$ precedes $j$, agent $i$ does not envy $j$ under $B$, while agent $j$ envies $i$ by at most one good under $B$.

We claim that the combined allocation $A\cup B$ is still \efone on the pair $(i,j)$. Since $i$ envies $j$ by at most one good under $A$, there exists a good $g\in A_j$ such that
\begin{align*}
v_i(A_i)\ge v_i(A_j\setminus\{g\}).
\end{align*}
Since $i$ does not envy $j$ under $B$, we also have
\begin{align*}
v_i(B_i)\ge v_i(B_j).
\end{align*}
By additivity,
\begin{align*}
v_i(A_i\cup B_i)\geq v_i((A_j\setminus\{g\})\cup B_j),
\end{align*}
and hence $i$ envies $j$ by at most one good under $A\cup B$. Similarly, since $j$ does not envy $i$ under $A$, and envies $i$ by at most one good under $B$, there exists a good $h\in B_i$ such that
\begin{align*}
v_j(A_j\cup B_j)\geq v_j(A_i\cup (B_i\setminus\{h\})),
\end{align*}
so $j$ also envies $i$ by at most one good under $A\cup B$. Therefore, the combined allocation remains \efone on the agents in the edge $e$.

Moreover, every good in $M_e$ is valued only by the agents of $e$. Therefore, allocating the goods in $M_e$ affects only the envy relations involving agents in $e$. Since hyperedges in the current color class are pairwise disjoint, these updates are independent across hyperedges in the class, and hence all hyperedges of the class can be processed in parallel without interference. It follows that, after processing the entire color class, the allocation remains \efone (although there may be envy cycles present in the envy graph). This completes the induction.

We now analyze the complexity. By \Cref{lemma:coloring}, computing the $(\Delta+1)$-coloring takes $O(\Delta\log^4m)$ depth and $O(\Delta m\log^2m)$ work. Fix a hyperedge $e$, and let $m_e:=|M_e|$. Since $|e|\le r=O(1)$, the envy-cycle elimination and the computation of a topological ordering both take $O(1)$ work and depth by \Cref{lemma:envycycle}. The constant-agent algorithm, applied to the agents of $e$ and the goods in $M_e$, takes $O(\log^2 m_e)$ depth and $O(m_e^{r+1})$ work by \Cref{thm:reachabilityreductionGen}. Since the hyperedges within a color class are processed independently in parallel, the work of one color class is $\sum_{e} O(m_e^{r+1})$, where the sum ranges over hyperedges in that class, while the depth is $O(\log^2 m)$.

Summing over all color classes, the total work is
\begin{align*}
    \sum_{e\in \mathcal{E}} O(m_e^{r+1}) \le O\left(\left(\sum_{e\in \mathcal{E}} m_e\right)^{r+1}\right)
    \;=\; O(m^{r+1}),
\end{align*}
since $\sum_{e\in\mathcal{E}} m_e = m$ and $r\ge 1$ is constant. There are $\Delta+1$ color classes, so the total depth during the allocation stage is $O(\Delta\log^2 m)$. This proves the theorem.
\end{proof}
\noindent 
This algorithm can be seen as a generalization of our algorithm for graphs; see \Cref{fig:hypergraph}.

\begin{figure}[t]
    \centering
    \begin{subfigure}[t]
        {0.49\textwidth}
        \centering
        \includegraphics[width=.9\linewidth]{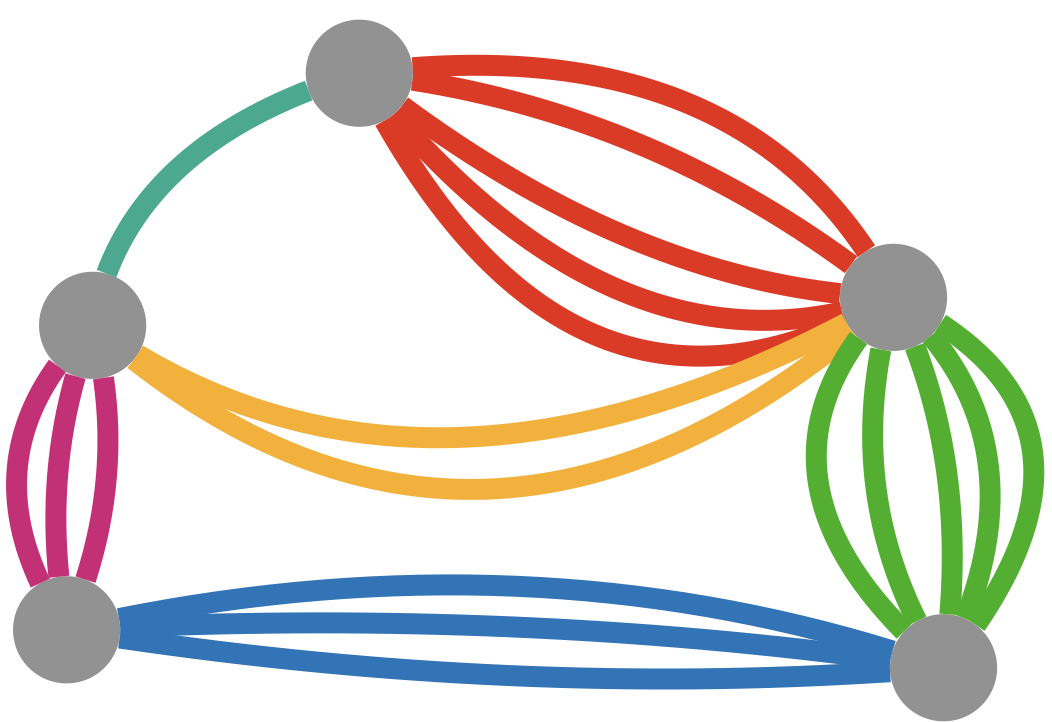}
        \subcaption{Graph $G$.}
        \label{subfig:graph}
    \end{subfigure}%
    ~
    \begin{subfigure}[t]
        {0.49\textwidth}
        \centering
        \includegraphics[width=.95\linewidth]{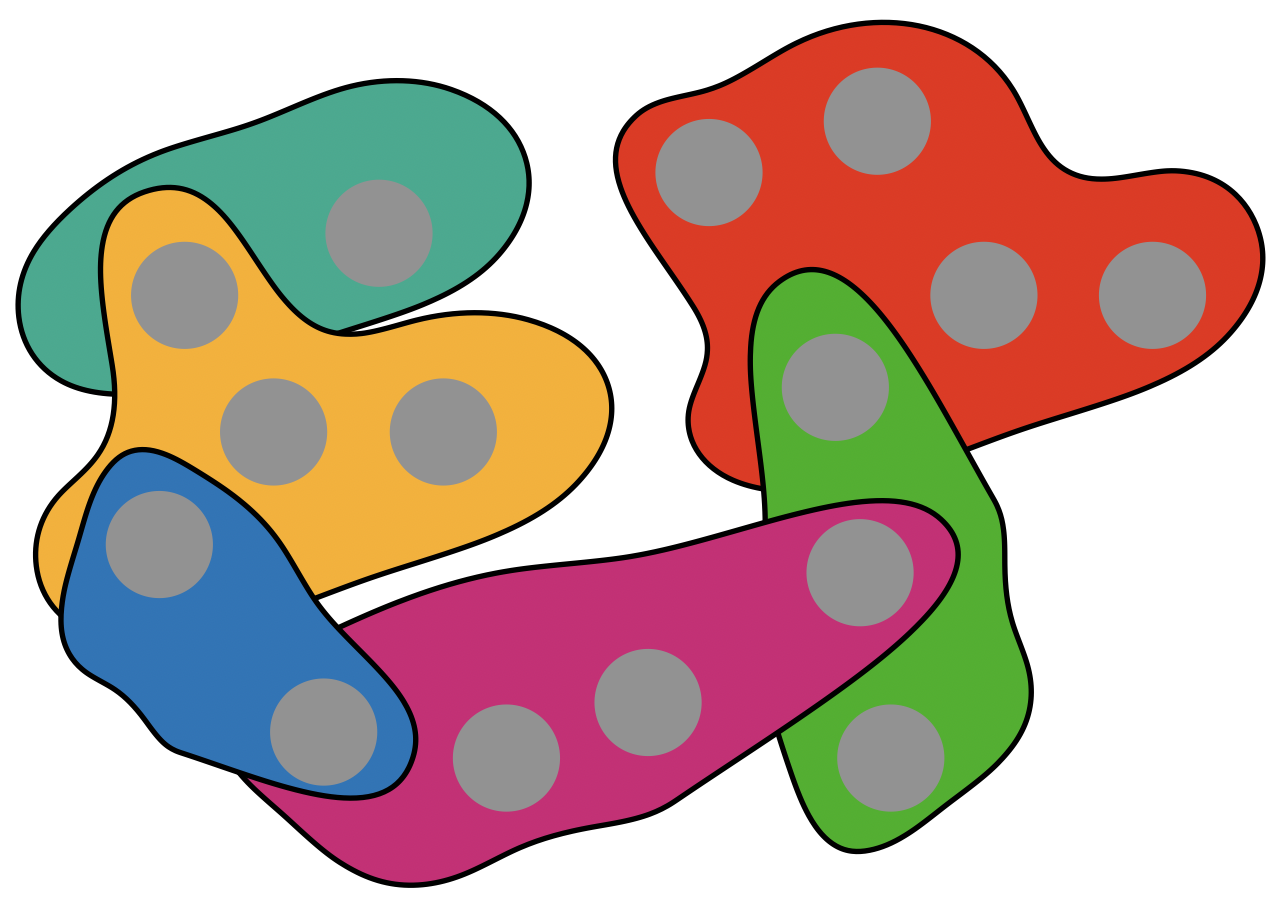}
        \subcaption{Hypergraph $H$.}
        \label{subfig:hypergraph}
    \end{subfigure}%
    \caption{A graph $G$ induced by instance $\cI:=\{M,N,\{v_i\}_i\}$ where vertices are agents and edges are goods (\Cref{subfig:graph}), and hypergraph $H$ induced by instance $\cI':=\{M',N',\{v'_i\}_i\}$ where vertices are agents and hyperedges are goods (\Cref{subfig:hypergraph}). Self-loops have been omitted for simplicity, and the hyperedges may correspond to multiedges. The algorithm for graphs (resp. hypergraphs) runs the algorithm for two (resp. $O(1)$) agents for each color class of (hyper)edges in parallel.}
    \label{fig:hypergraph}
\end{figure}

\subsection{Bypassing Fixed-Order Round Robin \CC-Hardness in \NC}
The \CC-hardness of \cite{garg2023fairly} for Fixed-Order Round Robin holds even when each good is positively valued by $3$ agents and each agent positively values at most $3$ goods. Equivalently, these are hypergraph instances with rank $3$ and maximum edge degree $6$. Thus, by giving up on simulating Fixed-Order Round Robin, our hypergraph algorithm bypasses their hardness result.
\bypassHard*
\begin{proof}    
This corresponds to a hypergraph instance with rank $r = O(1)$, and $\Delta = \polylog(m,n)$, and thus the proof follows by \Cref{thm:hypergraph_nc}.
\end{proof}




\section{Depth \texorpdfstring{$\tilde{O}(m/n)$}{tilde-O(m/n)} Algorithms with Polynomial Work}\label{sec:m/n}
In this section we give an algorithm for computing an \efone allocation with depth roughly $m/n$.
To obtain this we use the fact that computing a maximum matching is in \RNC \cite{Karp1986Constructing,Mulmuley1987Matching}.
\mnsmall*
\begin{proof}
    Given an allocation instance, construct a maximum matching instance as follows: create a complete bipartite graph $G$ with the agents on one side and the goods on the other. Specifically $V(G)=M\sqcup N$ and $E(G)$ contains edge $(a,g)$ for every agent $a\in N$ and good $j\in M$. The weight of edge $(a,g)$ between agent $a$ and good $g$ is $m-p_a(g)$, where $p_a(g)$ is the \textit{index} of good $g$ in agent $a$'s preference list (with ties broken arbitrarily). Let $\mathcal{M}\subseteq E(G)$ denote a maximum matching of this instance (so $|\mathcal{M}|=n$). 
    
    Then the \efone allocation algorithm iteratively computes a maximum matching $\cM_i$ in $G_i$ for $i\in[m/n]$ where $G_1=G$ and $G_i=G_{i-1}-(V(\cM_{i-1})\cap M)$ for each $i>1$. For each $\cM_i$ we obtain an allocation $\cA_i$ in which $\cA_i(a)=\{g\}$ iff $(a,g)\in\cM_i$ and return $\cA:=\bigcup_i\cA_i$. To solve the matching instances, we apply \Cref{thm:matching} to each $G_i$. 
    Furthermore, the algorithm of \Cref{thm:matching} is randomized with failure probability $1/2$, so we run it $c\log\frac{m}{n}$ times for any constant $c>1$ to reduce the failure probability to $n/m^c$ in each iteration. We say that the algorithm succeeds if none of the iterations failed to compute a maximum matching of their corresponding instance. As there are $m/n$ iterations, the probability that the algorithm succeeds is at least $1-1/m^{c-1}$ by a union bound.   
    Since $|V(G_i)|\leq |V(G_1)=|M|+|N|=m+n$ and $|E(G_i)|\leq |E(G_1)|=mn$ with $m\geq n$, then we have a parallel algorithm that runs in depth $O\left(\frac{m}{n}\log^2(m+ n)\cdot \log m\right)=O\left(\frac{m}{n}\log^3m\right)$ and work $O\left(\frac{m}{n}(mn)(m+n)^{3.5}\log^2m\cdot \log m \right)=O\left(m^{5.5}\log^3m\right)$. 

    It remains to show that $\cA$ is \efone if the algorithm succeeds. To do this, we prove that at any iteration $i$, the partial solution $\bigcup_{j\in[i]}\cA_j$ is \efone. Suppose this is not true; let $i$ be the earliest iteration such that $\bigcup_{j\in[i]}\cA_j$ is not \efone. Clearly $i>1$ because there must more than $1$ good allocated to each agent for the partial allocation to not be \efone. Furthermore, $\bigcup_{j\in[i]}\cA_j$ must be envy free up to 2 goods, since $\bigcup_{j\in[i-1]}\cA_j$ was \efone and $\cA_i$ only allocates one more good to each agent. Let $a$ be an agent that envies some other agent $b$ by $2$ goods $g,h$. Note that we must have $g\in\cM_{j}$ for some $j<i$ and $h\in\cM_i$, since iteration $i$ is the first iteration that isn't \efone. Let $f$ be the good that is matched to agent $a$ in $\cM_j$. Then if we had replaced the edge $(a,f)$ with $(a,h)$ in $\cM_{j}$, then we obtain a feasible matching (since $h$ was unmatched by iteration $i>j$) with strictly higher weight than $\cM_{j}$ (since good $h$ appears earlier than good $f$ in $a$'s preference list), contradicting the fact that $\cM_j$ is a maximum matching in $G_j$. 
\end{proof}

\noindent
Hence we have \RNC algorithms for \efone allocations for instances with $n\geq m/\polylog(m)$, and sublinear-depth algorithms for instances with $n\geq m^\eps$ for any constant $\eps>0$. 

\subsection{Bypassing Fixed-Order Round Robin \CC-Hardness in \RNC}
As mentioned above, Fixed-Order Round Robin is \CC-hard even when each good is positively valued by $3$ agents and each agent positively values at most $3$ goods. Applying our depth $\approx m/n$ algorithm, we are able to bypass this hardness even when we only have a polylog bound on the number of goods that each agent positively values.

\bypassHardAgain*
\begin{proof}
    We can assume without loss of generality that each good is positively valued by at least $1$ agent, since we can arbitrarily and efficiently allocate any good that is not positively valued by any agent to any agent. Then the number of goods $m$ is at most $\polylog(m,n)$. Running \Cref{thm:m/n-small} then gives an \RNC \efone allocation algorithm.
\end{proof}

\section{Algorithms for Relaxations of \efone}\label{sec:efK}

In this section we give parallel algorithms for \efk{\left(k\right)}, one with $k=\tilde{O}\left(\sqrt{m}\right)$ in RNC (\Cref{subsec:efrootm}) and one with $k=m^\eps$ in sublinear depth for any constant $\eps>0$ (\Cref{subsec:efm}). Both algorithms are based on partitioning the instance into disjoint instances and obtaining \efone allocations for each in parallel before combining them together. The second algorithm uses the algorithm with depth roughly $m/n$ that we gave in \Cref{sec:m/n}.

\subsection{\texorpdfstring{\efk{\left(\sqrt{m}\right)}}{efrootm} Allocations in \RNC}\label{subsec:efrootm}
Formally, the algorithm first arbitrarily partitions the $m$ goods into $\frac{m}{n}$ parts $\{M_i\}_{i\in[\frac{m}{n}]}$. We define an allocation instance with $N$ and $M_i$ for every $i$. Note $|M_i|=|N|=n$ for every $i$. Then for each instance $i$, we uniformly randomly sample a permutation over $n$ which corresponds to an allocation $\cA_i$ from the $n$ goods in $M_i$ to the $n$ agents. We return the allocation $\cA:=\bigcup_i \cA_i$.

Partitioning the instance can be done in constant depth, so the backbone of the computation is just sampling a random permutation over $n$.  
\noindent We have what we need to prove the following:
\efrootm*
\begin{proof}
    Fix a pair of agents $a,b$, one of the instances $i$, and the two goods $g^i,h^i\in M_i$ that were assigned to $a$ and $b$ in instance $i$. Crucially, we do not assume anything on which of $g^i$ and $h^i$ actually were assigned to $a$ and $b$; we only condition on these two goods being the goods that were assigned to either $a,b$ in instance $i$. Without loss of generality assume that $v_a(g^i)\geq v_a(h^i)$. By our fixing of the randomness and the fact that the assignment of goods to agents for instance $i$ was sampled uniformly randomly, observe that with probability $p\geq 1/2$ we have $a$ is assigned $g^i$ and $b$ is assigned $h^i$, and vice-versa. We can fix the randomness and define $g^i$ and $h^i$ in the same way for each instance $i$, where the same property holds. For an instance $i$, we say that $a$ received its preferred good if $a$ was assigned $g^i$.
    
    For the sake of analysis, sort the instances according to $V_a(i):=v_a(g^i)-v_a(h^i)$. That is, we reorder the instances such that instance $1$ maximizes $v_a(g^i)-v_a(h^i)$ over all $i$, instance $2$ maximizes $v_a(g^i)-v_a(h^i)$ over the remaining instances, and so on. With this re-sorting of the instances, we define the following random walk on $\mathbb{Z}$ with respect to $a$: the walk starts at $K:=\left\lceil \sqrt{\frac{3m\log m}{n}}\right\rceil$ and takes a $+1$ step at time $t$ if $a$ received its preferred good in instance $t$ (where we use the new sorted order of the instances) and a $-1$ step otherwise. This is equivalent to an $\frac{m}{n}$-step random walk starting at $K$ where $a$ takes a $+1$ step with probability $p$ and a $-1$ step with probability $1-p$ at each time $t$. Applying \Cref{lemma:hoeffding} with $N=m/n,c=3,d=n$ to this random walk, we have with probability at most $1/m^3$ that the random walk never goes below $0$.
    
    We next prove that if the random walk never goes below $0$, then $a$ envies $b$ by at most $K$ goods. Define $\pm1$ random variables $X_t$ for each $t\in\left[\frac{m}{n}\right]$ where $X_t$ is the step of the random walk at time $t$ and let $X=\sum_t X_t$. Intuitively, the claim follows from the way we ordered the instances, which defined the order of the steps in the random walk. Note that for the random walk to never step below $0$, we must have $\sum_{t\in[T]}X_t\geq-K$ for any $T$. In other words, for any $T$, if we ignore the first $K$ of the $-1$ steps by time $T$, then there is an injective mapping $\cM$ of the remaining $-1$ steps by time $T$ to the $+1$ steps by time $T$ such that each $-1$ step that occurred at time $t$ is mapped to a $+1$ step that occurred at time $t'<t$. For a time step $i$ let $\cM(i)$ denote its mapping by $\cM$ and let $\operatorname{dom}(\cM)$ denote the domain of $\cM$, i.e.\ the set of $-1$ steps $i$ that were mapped to an earlier $+1$ step by $\cM$. So we have that $a$ received $h^i$ in instance $i$ and $a$ received $g^j$ in instance $\cM(i)$ where $\cM(i)<i$.
    Then by our re-sorting of the instances, we have $\sum_{i\in \operatorname{dom}(\cM)}v_a(g^{\cM(i)})-v_a(h^{\cM(i)}) \geq \sum_{i\in\operatorname{dom}(\cM)}v_a(g^i)-v_a(h^i)$, where agent $a$ received all of the $g^{\cM(i)}$'s and all of the $h^i$'s in these two sums. The only instances we have not considered yet were the $K$ ignored instances that correspond to the first $K$ of the $-1$ steps in the random walk; these correspond to a subset $S\subseteq \cA(b)$ of $K$ goods, each of the form $g^i$ for some instance $i$. Then we have 
    \begin{align*}
        v_a\left(\cA(a)\right) = \sum_{i\in \operatorname{dom}(\cM)}v_a(g^{\cM(i)})-v_a(h^{\cM(i)}) \geq \sum_{i\in\operatorname{dom}(\cM)}v_a(g^i)-v_a(h^i)= v_a\left(\cA(b)\setminus S\right)
    \end{align*}
    proving the claim that $a$ envies $b$ by at most $K$ goods if the random walk never went below $0$. Observe that an equivalent claim to the previous claim is that if $|X|<K$ then $a$ envies $b$ by at most $K$ goods. This is because the event that $a$ goes below $0$ (having started at $K$) is the event that $X<K$ which is a subset of the event that $|X|\geq K$. Then with probability at most $1/m^3$ agent $a$ envies $b$ by more than $\sqrt{\frac{3m\log m}{n}}$ goods. After union bounding over all ordered pairs of agents we have that $\cA$ is \efk{\left(\sqrt{\frac{3m\log m}{n}}\right)} with probability at least $1-1/m$.

    The algorithm's runtime is that of sampling a uniformly random permutation over $n$ for each of the $m/n$ instances so by \Cref{lemma:permutation} its depth is $O(\log n)$ and its work is $O(n)\cdot m/n=O(m)$. 
\end{proof}

\subsection{\texorpdfstring{\efk{\left(m^\eps\right)}}{ekmeps} Allocations with Sublinear Depth}\label{subsec:efm}
The algorithm for computing \efk{\left(m^\eps\right)} allocations is similar to the \efk{\left(\sqrt{m}\right)} algorithm in that it breaks up the instance into disjoint instances and solves each in parallel. The difference is that it partitions the instances such that in each instance the number goods is \textit{close} to the number of agents, but not equal. Hence, we use our depth $m/n$ algorithm.

Given some $k\in[m/n]$, the algorithm first arbitrarily partitions the $m$ goods into $\frac{m}{nk}$ parts $\{M_i\}_{i\in[\frac{m}{nk}]}$. We define an allocation instance with $N$ and $M_i$ for every $i$. Then for each instance $i$, we run the algorithm of \Cref{thm:m/n-small} to obtain an \efone allocation $\cA_i$. We return the allocation $\cA:=\bigcup_i \cA_i$. 
\efm*
\begin{proof}
    We use the above algorithm. Fix any pair of agents $a,b$. By \Cref{thm:m/n-small}, $a$ envies $b$ by at most $1$ good $g_i\in M_i$ for each instance $\cA_i$ (note if $a$ doesn't envy $b$ then $g_i$ can be any good). Since there are $\frac{m}{nk}$ instances and the $M_i$'s are disjoint, there exists a subset $S\subseteq \cA(b),|S|\leq\frac{m}{nk}$ of goods such that 
    \begin{align*}
        v_a(\cA(a)) = \sum_i v_a\left( \cA_i(a)\right)\geq \sum_i v_a\left( \cA_i(b)\setminus g_i\right) = v_a(\cA(b)\setminus S)
    \end{align*}
    Hence $\cA$ is an \efk{\left(\frac{m}{nk}\right)} allocation.
    By \Cref{thm:m/n-small}, the algorithm's depth is $O(nk\log^3 (nk)/n) = O\left(k\log^2(nk)\right)$ and its work is $O\left((nk)^{5.5}\log^3(nk)\right)$. The algorithm from \Cref{thm:m/n-small} can be made to fail with probability at most $1/m^2$ (by setting $c=3$ in the proof of \Cref{thm:m/n-small}); by union bounding over the $\frac{m}{nk}$ parts we have that the probability that the algorithm succeeds for all of the parts is at least $1-1/m$.
    The theorem follows by setting $k=\Theta\left(m^{1-\eps}\right)$.
\end{proof}

\section{\CC-Completeness of Fixed-Order Round Robin}
We show that computing Fixed-Order Round Robin \efone allocations is in \CC by a reduction to the stable matching problem, a canonical problem for the class \CC.

\begin{definition}[Stable Matching]~\label{def:stable_matching}
    In the stable matching problem with complete preferences, we are given a bipartite graph with parts $A$ and $B$, where each vertex in $A$ has a strict preference ordering over all vertices in $B$ and vice-versa. The goal is to compute a matching $M$ such that there is no unmatched pair $a\in A,b\in B$ for which $a$ prefers $b$ to its partner in $M$ and $b$ prefers $a$ to its partner in $M$. Such a matching is called \defn{stable}.
\end{definition}
\begin{theorem}[\cite{greenlaw1995limits}]\label{thm:stable_matching_cc_hard}
    Stable matching is \CC-complete.
\end{theorem}
\noindent
\cite{garg2023fairly} showed that computing Fixed-Order Round-Robin allocations is \CC-hard even on instances where each agent positively values at most $3$ goods and each good is positively valued by at most $3$ agents.
\begin{theorem}[\cite{garg2023fairly}]\label{thm:EF1-CC-hard}
    Fixed-Order Round Robin is \CC-hard.
\end{theorem}
\noindent
We now show that Fixed-Order Round Robin is in \CC, and thus, it follows that computing a Fixed-Order Round Robin \efone allocation is \CC-Complete.

\cccomplete*
\begin{proof}
    We make use of stable matching as defined above.
    
    \textbf{Reduction.} Given an allocation instance with $m$ goods, $n$ agents with their (possibly incomplete) preferences $v$ on the goods, and a permutation $\sigma$ on the agents, we construct an instance of the stable matching problem as follows: let $A$ contain $m/n$ (numbered) copies of each agent (thus, $|A|=m$), and $B$ contain the set of goods. The preference order for each agent copy is the same as the preference order for that agent on the good set; if an agent's preference list is incomplete then we arbitrarily order the remaining goods in its ordering of the goods. The preference order for each good is the permutation $\sigma$ repeated $m/n$ times, where the $i$th appearance of an agent corresponds to the $i$th copy of that agent; hence $B$ has identical preference lists. It is easy to verify that the entire reduction can be carried out in logarithmic space.
    
    \textbf{Correctness.} 
    We first show that any stable matching instance where all vertices on one side have identical preference lists has a unique solution. Let $a$ be the first vertex in this common preference list. In any stable matching, $a$ must be matched to its most preferred partner. Otherwise, if $b$ is its most preferred partner, then $b$ is matched to some other vertex that appears later than $a$ in the common preference list, and hence $a$ and $b$ form a blocking pair. Removing $a$ and its matched partner, the same argument applies to the remaining instance. By induction, the stable matching is unique, and is obtained by considering the vertices in the common preference order and assigning each vertex its most preferred unmatched partner.

    Applying this to our constructed instance, we obtain that the unique stable matching is formed by first matching the first copy of each agent in the order $\sigma$, then the second copy of each agent in the same order, and so on, where each copy is matched to its most preferred unmatched good. But this is exactly the allocation produced by Fixed-Order Round Robin: i.e., in each round, the agents choose in the order $\sigma$, and each agent receives its most preferred good among those that remain.


    
    Hence, by \Cref{thm:stable_matching_cc_hard}, Fixed-Order Round Robin is in \CC. Combining with the fact that Fixed-Order Round Robin is also \CC-Hard by \Cref{thm:EF1-CC-hard}, we then have that Fixed-Order Round Robin is \CC-Complete.
\end{proof}
\noindent We remark that there exists a parallel algorithm for stable matching whose depth is sublinear in the sum of all preference lists \cite{Feder2000ASublinear}.
Unfortunately, our reduction from Fixed-Order Round Robin to stable matching constructs an instance of stable matching where all agents have complete preferences and both sides of the bipartite graph have size $m$, so the sum of preference lists is $O(m^2)$. Running the parallel stable matching algorithm gives a $\sqrt{m^2}=m$ depth algorithm, which is not interesting.

\section{Conclusion}
We extended several results from \cite{garg2023fairly} and showed that \efone allocations can be computed efficiently in parallel for a much larger class of instances than what was previously understood. The main open question that remains is the following:
\begin{quote} \centering
    \textit{Do algorithms for \efone allocations with sublinear depth and polynomial work\\ exist for all instances with additive valuations?}
\end{quote}
\noindent
Even a parallel algorithm with polynomial work whose depth is sublinear on the total sum of preference list lengths would be nontrivial, and would be analogous to the sublinear stable matching algorithm of \cite{Feder2000ASublinear}.  

Also, recall that \Cref{thm:reachabilityreduction} and \Cref{thm:m/n-small} together give \NC algorithms when the number of agents is $O(1)$ or \RNC algorithms when the number of agents is $\Omega(m/\polylog(m))$. Hence, we have shown that \efone allocations can be computed efficiently in parallel when the number of agents is small enough or large enough. A natural next goal in addressing  the above question is to give efficient parallel algorithms for \efone allocations for an intermediate number of agents; for instance, can \efone allocations be found in \NC for $n = \polylog(m)$ agents with additive valuations? 

\bibliographystyle{alpha}
\bibliography{references}

@string{teac = "ACM Transactions on Economics and Computation (TEAC)"}

@string{jacm = "Journal of the ACM (JACM)"}

@string{siamjc = "SIAM Journal on Computing (SICOMP)"}

@string{siamdm = "SIAM Journal on Discrete Mathematics (SIDMA)"}

@string{tcs = "Theoretical Computer Science (TCS)"}

@string{tocs = "Theory of Computing Systems (TOCS)"}

@string{ipl = "Information Processing Letters (IPL)"}

@string{jcss = "Journal of Computer and System Sciences (JCSS)"}

@string{toct = "ACM Transactions on Computation Theory (TOCT)"}

@string{or = "Operations Research"}

@string{ai = "Artificial Intelligence"}

@string{mlq = "Mathematical Logic Quarterly (MLQ)"}

@string{combinatorica = "Combinatorica"}

@string{jpe = "Journal of Political Economy"}

@string{cambridge = "Mathematical Proceedings of the Cambridge Philosophical Society"}

@string{spaa = "ACM Symposium on Parallelism in Algorithms and Architectures (SPAA)"}

@string{focs = "IEEE Symposium on Foundations of Computer Science (FOCS)"}

@string{ipdps = "IEEE International Parallel and Distributed Processing Symposium (IPDPS)"}

@string{opodis = "International Conference on Principles of Distributed Systems (OPODIS)"}

@string{aamas = "International Conference on Autonomous Agents and Multiagent Systems (AAMAS)"}

@string{ec = "ACM Conference on Economics and Computation (EC)"}

@string{wine = "Conference on Web and Internet Economics (WINE)"}

@string{fsttcs = "Foundations of Software Technology and Theoretical Computer Science (FSTTCS)"}

@string{aaai = "AAAI Conference on Artificial Intelligence (AAAI)"}

@string{ijcai = "International Joint Conference on Artificial Intelligence (IJCAI)"}

@string{adt = "Algorithmic Decision Theory (ADT)"}

@string{ecai = "European Conference on Artificial Intelligence (ECAI)"}

@inproceedings{garg2023fairly,
  title={Fairly Allocating Goods in Parallel},
  author={Garg, Rohan and Psomas, Alexandros},
  booktitle=aamas,
  year={2025}
}

@article{Feder2000ASublinear,
	author = {Tom{\'a}s Feder and Nimrod Megiddo and Serge A. Plotkin},
	issn = {0304-3975},
	journal = tcs,
	number = {1},
	pages = {297-308},
	title = {A sublinear parallel algorithm for stable matching},
	volume = {233},
	year = {2000},
}

@article{Karp1986Constructing,
	author = {Karp, R. M. and Upfal, E. and Wigderson, A.},
	doi = {10.1007/BF02579407},
	id = {Karp1986},
	isbn = {1439-6912},
	journal = combinatorica,
	number = {1},
	pages = {35--48},
	title = {{Constructing a Perfect Matching is in Random NC}},
	volume = {6},
	year = {1986},
}

@article{Mulmuley1987Matching,
	author = {Mulmuley, Ketan and Vazirani, Umesh V. and Vazirani, Vijay V.},
	id = {Mulmuley1987},
	isbn = {1439-6912},
	journal = combinatorica,
	number = {1},
	pages = {105--113},
	title = {{Matching is as Easy as Matrix Inversion}},
	volume = {7},
	year = {1987},
}

@article{lee2017apxnsw,
  title={APX-hardness of maximizing Nash social welfare with indivisible items},
  author={Lee, Euiwoong},
  journal=ipl,
  volume={122},
  pages={17--20},
  year={2017},
  publisher={Elsevier}
}

@book{greenlaw1995limits,
  title={Limits to parallel computation: P-completeness theory},
  author={Greenlaw, Raymond and Hoover, H James and Ruzzo, Walter L},
  year={1995},
  publisher={Oxford university press}
}

@book{jaja1997introduction,
	author = {JaJa, J.},
	publisher = {Addison Wesley},
	title = {An Introduction to Parallel Algorithms},
	year = {1997}
}

@inproceedings{blelloch2020optimal,
    author = {Blelloch, Guy E. and Fineman, Jeremy T. and Gu, Yan and Sun, Yihan},
    title = {Optimal Parallel Algorithms in the Binary-Forking Model},
    year = {2020},
    booktitle = spaa
}

@book{CLRS,
	author    = {Thomas H. Cormen and
	Charles E. Leiserson and
	Ronald L. Rivest and
	Clifford Stein},
	title     = {Introduction to Algorithms},
	publisher = {MIT Press},
	year      = {2009},
}

@Article{ABP01,
  author="Arora, N. S.
  and Blumofe, R. D.
  and Plaxton, C. G.",
  title="Thread Scheduling for Multiprogrammed Multiprocessors ",
  journal=tocs,
  year="2001",
  volume="34",
  number="2",
}

@article{BL98,
  author    = {Robert D. Blumofe and
  Charles E. Leiserson},
  title     = {Space-Efficient Scheduling of Multithreaded Computations},
  journal   = siamjc,
  volume    = {27},
  number    = {1},
  year      = {1998},
}

@inproceedings{fitzsimmons2025parallelizability,
  author       = {Zack Fitzsimmons and
                  Zohair Raza Hassan and
                  Edith Hemaspaandra},
  title        = {On the Parallelizability of Approval-Based Committee Rules},
  booktitle    = ecai,
  year         = {2025}
}

@inproceedings{mahara2025polynomial,
  title     = {A Polynomial-Time Algorithm for Fair and Efficient Allocation with a Fixed Number of Agents},
  author    = {Mahara, Ryoga},
  booktitle = wine,
  year      = {2025}
}

@article{budish11combinatorial,
    author = {Budish, Eric},
    title = {The Combinatorial Assignment Problem: Approximate Competitive Equilibrium from Equal Incomes},
    journal = jpe,
    volume = {119},
    number = {6},
    pages = {1061-1103},
    year = {2011}
}

@article{oh21fairly,
  author       = {Hoon Oh and
                  Ariel D. Procaccia and
                  Warut Suksompong},
  title        = {Fairly Allocating Many Goods with Few Queries},
  journal      = siamjc,
  volume       = {35},
  number       = {2},
  pages        = {788--813},
  year         = {2021}
}

@inproceedings{lipton04approximately,
  author       = {Richard J. Lipton and
                  Evangelos Markakis and
                  Elchanan Mossel and
                  Amin Saberi},
  title        = {On approximately fair allocations of indivisible goods},
  booktitle    = ec,
  pages        = {125--131},
  year         = {2004}
}

@inproceedings{feige25low,
  author       = {Uriel Feige},
  title        = {Low communication protocols for fair allocation of indivisible goods},
  booktitle    = ec,
  pages        = {358--382},
  year         = {2025}
}

@book{foley66resource,
  title={Resource allocation and the public sector},
  author={Foley, Duncan Karl},
  year={1966},
  publisher={Yale University}
}

@article{caragiannis19unreasonable,
  author       = {Ioannis Caragiannis and
                  David Kurokawa and
                  Herv{\'{e}} Moulin and
                  Ariel D. Procaccia and
                  Nisarg Shah and
                  Junxing Wang},
  title        = {The Unreasonable Fairness of Maximum Nash Welfare},
  journal      = teac,
  volume       = {7},
  number       = {3},
  pages        = {12:1--12:32},
  year         = {2019},
}

@inproceedings{barman18finding,
  author       = {Siddharth Barman and
                  Sanath Kumar Krishnamurthy and
                  Rohit Vaish},
  title        = {Finding Fair and Efficient Allocations},
  booktitle    = ec,
  year         = {2018},
}

@inproceedings{christodoulou23fair,
  author       = {George Christodoulou and
                  Amos Fiat and
                  Elias Koutsoupias and
                  Alkmini Sgouritsa},
  title        = {Fair allocation in graphs},
  booktitle    = ec,
  publisher    = {{ACM}},
  year         = {2023}
}

@article{chaudhury24efx,
  author       = {Bhaskar Ray Chaudhury and
                  Jugal Garg and
                  Kurt Mehlhorn},
  title        = {{EFX} Exists for Three Agents},
  journal      = jacm,
  volume       = {71},
  number       = {1},
  pages        = {4:1--4:27},
  year         = {2024}
}

@article{plaut20almost,
  author       = {Benjamin Plaut and
                  Tim Roughgarden},
  title        = {Almost Envy-Freeness with General Valuations},
  journal      = siamdm,
  volume       = {34},
  number       = {2},
  pages        = {1039--1068},
  year         = {2020}
}

@article{mayr92complexity,
  author       = {Ernst W. Mayr and
                  Ashok Subramanian},
  title        = {The Complexity of Circuit Value and Network Stability},
  journal      = jcss,
  volume       = {44},
  number       = {2},
  pages        = {302--323},
  year         = {1992}
}

@article{cook14complexity,
  author       = {Stephen A. Cook and
                  Yuval Filmus and
                  Dai Tri Man Le},
  title        = {The complexity of the comparator circuit value problem},
  journal      = toct,
  volume       = {6},
  number       = {4},
  pages        = {15:1--15:44},
  year         = {2014}
}

@book{brandt16handbook,
  editor       = {Felix Brandt and
                  Vincent Conitzer and
                  Ulle Endriss and
                  J{\'{e}}r{\^{o}}me Lang and
                  Ariel D. Procaccia},
  title        = {Handbook of Computational Social Choice},
  publisher    = {Cambridge University Press},
  year         = {2016}
}

@article{othman16complexity,
  author       = {Abraham Othman and
                  Christos H. Papadimitriou and
                  Aviad Rubinstein},
  title        = {The Complexity of Fairness Through Equilibrium},
  journal      = teac,
  volume       = {4},
  number       = {4},
  pages        = {20:1--20:19},
  year         = {2016}
}

@inproceedings{bu24logarithmic,
  author       = {Xiaolin Bu and
                  Zihao Li and
                  Shengxin Liu and
                  Jiaxin Song and
                  Biaoshuai Tao},
  title        = {Logarithmic Comparison-Based Query Complexity for Fair Division of
                  Indivisible Goods},
  year         =  {2024},
  booktitle    = wine
}

@article{aziz24best,
  author       = {Haris Aziz and
                  Rupert Freeman and
                  Nisarg Shah and
                  Rohit Vaish},
  title        = {Best of Both Worlds: Ex Ante and Ex Post Fairness in Resource Allocation},
  journal      = or,
  volume       = {72},
  number       = {4},
  pages        = {1674--1688},
  year         = {2024}
}

@article{amanatidis23fair,
  author       = {Georgios Amanatidis and
                  Haris Aziz and
                  Georgios Birmpas and
                  Aris Filos{-}Ratsikas and
                  Bo Li and
                  Herv{\'{e}} Moulin and
                  Alexandros A. Voudouris and
                  Xiaowei Wu},
  title        = {Fair division of indivisible goods: Recent progress and open questions},
  journal      = ai,
  volume       = {322},
  pages        = {103965},
  year         = {2023}
}

@inproceedings{afshinmehr25efx,
  author       = {Mahyar Afshinmehr and
                  Alireza Danaei and
                  Mehrafarin Kazemi and
                  Kurt Mehlhorn and
                  Nidhi Rathi},
  title        = {{EFX} Allocations and Orientations on Bipartite Multi-graphs: {A} Complete Picture},
  booktitle    = aamas,
  year         = {2025}
}

@inproceedings{sgouritsa25existence,
  author       = {Alkmini Sgouritsa and
                  Minas Marios Sotiriou},
  title        = {On the Existence of {EFX} Allocations in Multigraphs},
  booktitle    = aamas,
  year         = {2025}
}

@inproceedings{bhaskar25extending,
  author       = {Umang Bhaskar and
                  Yeshwant Pandit},
  title        = {Extending {EFX} Allocations to Further Multi-Graph Classes},
  booktitle    = fsttcs,
  year         = {2025}
}

@inproceedings{zheng19parallel,
  author       = {Xiong Zheng and
                  Vijay K. Garg},
  title        = {Parallel and Distributed Algorithms for the Housing Allocation Problem},
  booktitle    = opodis,
  year         = {2019}
}

@book{subramanian1989new,
  title={A new approach to stable matching problems},
  author={Subramanian, Ashok},
  year={1989},
  publisher={Stanford University}
}

@incollection{gazit1988optimal,
    title={Optimal tree contraction in the EREW model},
    author={Gazit, Hillel and Miller, Gary L and Teng, Shang-Hua},
    booktitle={Concurrent Computations: Algorithms, Architecture, and Technology},
    pages={139--156},
    year={1988},
    publisher={Springer}
}

@inproceedings{miller1985parallel,
    title={Parallel tree contraction and its application},
    author={Miller, Gary L and Reif, John H},
    booktitle=focs,
    volume={26},
    pages={478--489},
    year={1985}
}

@inproceedings{Misra2025EFGraph,
	address = {Cham},
	author = {Misra, Neeldhara and Sethia, Aditi},
	booktitle = adt,
	pages = {258--272},
	title = {Envy-Free and Efficient Allocations for Graphical Valuations},
	year = {2025}
}

@inproceedings{csar17winner,
  author       = {Theresa Csar and
                  Martin Lackner and
                  Reinhard Pichler and
                  Emanuel Sallinger},
  title        = {Winner Determination in Huge Elections with MapReduce},
  booktitle    = aaai,
  year         = {2017}
}

@inproceedings{csar18computing,
  author       = {Theresa Csar and
                  Martin Lackner and
                  Reinhard Pichler},
  title        = {Computing the Schulze Method for Large-Scale Preference Data Sets},
  booktitle    = ijcai,
  year         = {2018}
}

@inproceedings{hu20nc,
  author       = {Changyong Hu and
                  Vijay K. Garg},
  title        = {{NC} Algorithms for Popular Matchings in One-Sided Preference Systems
                  and Related Problems},
  booktitle    = ipdps,
  pages        = {759--768},
  year         = {2020}
}

@article{brandt09computational,
  author       = {Felix Brandt and
                  Felix A. Fischer and
                  Paul Harrenstein},
  title        = {The Computational Complexity of Choice Sets},
  journal      = mlq,
  volume       = {55},
  number       = {4},
  pages        = {444--459},
  year         = {2009}
}

@article{goldberg1989ParMIS,
	author = {Goldberg, Mark and Spencer, Thomas},
	doi = {10.1137/0218029},
	journal = siamjc,
	number = {2},
	pages = {419-427},
	title = {A New Parallel Algorithm for the Maximal Independent Set Problem},
	volume = {18},
	year = {1989},
}

\end{document}